%% file: arXiv_v1.tex
\documentclass[aps,prl,reprint,superscriptaddress,nofootinbib]{revtex4-2}

\makeatletter
\AtBeginDocument{\let\auto@bib@innerbib\@empty}
\makeatother

\usepackage{amsmath,amsthm,amssymb,mathtools,bm}
\usepackage{color}

\usepackage{tikz}
\usetikzlibrary{arrows.meta,positioning,calc,decorations.pathreplacing}
\usepackage[
  colorlinks=false,
  breaklinks=true,
  pdfborder={0 0 1}
]{hyperref}

\definecolor{aqorange}{RGB}{205,120,25}
\definecolor{imblue}{RGB}{55,105,200}
\definecolor{parblue}{RGB}{65,125,205}

\usetikzlibrary{decorations.pathreplacing}

\newcommand{\ket}[1]{|#1\rangle}
\newcommand{\bra}[1]{\langle #1|}
\newcommand{\proj}[1]{|#1\rangle\langle #1|}

\theoremstyle{definition}

\newtheorem{proposition}{Proposition}
\newtheorem{lemma}{Lemma}
\newtheorem{theorem}{Theorem}
\newtheorem{corollary}{Corollary}
\newtheorem{definition}{Definition}

\begin{document}

\title{Resource Theory of Aquaternionicity: Every Pure State Is a Replicable Resource for Exact Simulation of Arbitrary Quantum Operations}


\author{Hayato Arai}
\email{h.arai6626@gmail.com, haya.arai@ntt.com}
\affiliation{NTT Communication Science Laboratories, NTT, Inc., 3-1 Morinosato Wakamiya, Atsugi, Kanagawa 243-0198, Japan}

\author{Yasuaki Nakayama}
\email{yasuaki.nakayama@ntt.com}
\affiliation{NTT Communication Science Laboratories, NTT, Inc., 3-1 Morinosato Wakamiya, Atsugi, Kanagawa 243-0198, Japan}

\begin{abstract}
An \emph{alchemical resource theory} admits a nonfree state that can be exactly replicated by free operations and whose finitely many copies enable exact simulation of arbitrary quantum instruments. The only known nontrivial examples, imaginarity and parity asymmetry, exhaust the qubit case (arXiv:2609.12988) and exhibit an all-or-nothing property: any nonmaximal resource exactly simulates only free unitaries. We introduce \emph{aquaternionicity} {for even-dimensional qudits \(d\ge4\)}, based on a quaternionic structure represented by an antiunitary \(\Theta\) with \(\Theta^2=-I\). It is extremal in two senses: every pure state is alchemical, and alchemicality is maximally abundant under mixing. {Within this framework,} for any nontrivial \(d\)-dimensional resource theory and pure alchemical state \(g\), the admissible pure-mixing set satisfies \(\dim\mathcal M_g\le 2d-4,\) which aquaternionicity saturates for every pure state. Conversely, these two extremal properties characterize the quaternionic structure. Because every pure state is alchemical, the initial resource can be chosen to match the known noise structure. Moreover, when the noisy resource states remain within a {common isotropic subspace}, the same free processor works exactly for all of them while returning the resource unchanged. This enables recalibration-free exact simulation, including universal quantum computation, and is realized for standard dephasing and dissipative noise, where aquaternionicity substantially outperforms imaginarity and parity asymmetry.
\end{abstract}

\maketitle


In quantum resource theories, nonfree states can enable operations that are otherwise restricted by free operations \cite{ChitambarGour2019}. Magic-state quantum computation provides a paradigmatic example, although magic states are typically consumed in implementing non-Clifford operations \cite{BravyiKitaev2005}. Here we consider a stronger task: a nonfree state is required to be exactly replicable by free operations and, together with suitable free operations, to enable the exact simulation of arbitrary quantum instruments. Such a state provides a replenishable resource for universal quantum information processing{~\cite{Aberg2014,GarciaDiaz2018,LostaglioMuller2019,MarvianSpekkens2019,KuroiwaYamasaki2020,ZhangLiLuo2024,SonEtAl2026}}. This capability was formalized as \emph{alchemical resource theory} in the inverse-problem approach of Ref.~\cite{InverseDraft}. Remarkably, for qubits this inverse problem admits only two nontrivial theories, up to local changes of basis: imaginarity \cite{HickeyGour2018,Wu2021} and parity asymmetry \cite{GourSpekkens2008,MarvianSpekkens2014,InverseDraft}.

Exact simulation in these theories exhibits an all-or-nothing property {\cite{MarvianSpekkensWAY2012,NakayamaTakeuchiAkibue2026,HokkyoTajima2026,InverseDraft}}: any nonmaximal resource state can exactly simulate only free unitaries. Thus, exact simulation of nonfree dynamics requires the resource to remain in the alchemical class. This raises a broader question for \emph{qudits}: \emph{how broadly can replicable resources for exact universal simulation populate the quantum state space?} We address this from two complementary perspectives: which pure states can be alchemical, and how much mixing around a pure alchemical resource \(g\) can preserve alchemicality. For the latter, let \(\mathcal M_g\) be the set of pure states \(h\) such that every nontrivial mixture \((1-t)g+th\), \(0<t<1\), remains alchemical. For qudits with local dimension \(d\), we establish the universal bound
\begin{equation}
 \dim\mathcal M_g\le 2d-4.
 \label{eq:intro-bound}
\end{equation}
{For even \(d\ge4\), neither imaginarity nor parity asymmetry attains this bound.}

A theory attaining both extrema emerges naturally by focusing on the Jordan-algebraic structure underlying quantum state spaces. Imaginarity may be viewed as an asymmetry resource theory within complex quantum mechanics whose set of free states realizes real quantum theory, \(H_d(\mathbb R)\subset H_d(\mathbb C)\) \cite{HickeyGour2018}. In the Euclidean Jordan-algebraic description of quantum state spaces, alongside the real and complex Hermitian matrix algebras underlying real and complex quantum theory, one also finds the quaternionic Hermitian algebra \(H_n(\mathbb H)\) {\cite{Jordan1934,BarnumWilce2014,BarnumGraydonWilce2020,SonodaAraiHayashi2025,Finkelstein1962,Adler1995}}. This suggests replacing the real structure underlying imaginarity by its quaternionic counterpart. The algebra \(H_n(\mathbb H)\) has a standard representation on a \(2n\)-dimensional complex Hilbert space as the operators invariant under an antiunitary \(\Theta\) satisfying \(\Theta^2=-I\) {\cite{Dyson1962,Baez2012,Graydon2013,PinedaMontoyaRuaMunoz2026}}.
{Antiunitary symmetries have previously been treated resource-theoretically in time-reversal frameness and special antiunitary asymmetry \cite{GourSandersTurner2009,BuFanPatiWu2016}, and the spinorial time-reversal setting was already noted to make every pure state resourceful \cite{GourSandersTurner2009}; more recently, time-reversal symmetry breaking has also been formulated as a dynamical resource theory \cite{MiyazakiKuroiwaMurao2026}.}
This motivates \emph{aquaternionicity}, the asymmetry resource theory whose free states satisfy \(\Theta\rho\Theta^{-1}=\rho\), with the corresponding \(\Theta\)-covariant free operations. For \(m\) elementary systems, \(\Theta_m=\Theta^{\otimes m}\) and \(\Theta_m^2=(-1)^mI\), yielding alternating quaternionic and real structures for odd and even composites \cite{BarnumGraydonWilce2020}.

Aquaternionicity provides an extremal answer to the question above. Most strikingly, \emph{every pure state is a replicable resource for the exact simulation of arbitrary quantum operations}. {As shown below, combining the support criterion for alchemicality with Kramers orthogonality \cite{FruchartYaoVitelli2023}, \(\langle\Theta\psi|\psi\rangle=0\), valid for every \(\ket{\psi}\) when \(\Theta^2=-I\), establishes this statement.} Moreover, aquaternionicity attains equality in Eq.~\eqref{eq:intro-bound} for every pure resource. Alchemical resources are not restricted to pure states: any state whose support is orthogonal to its \(\Theta\)-image retains the same exact simulation capability. Aquaternionicity also inherits operational capabilities established for imaginarity, including resource distillation, assisted state conversion, and channel discrimination {\cite{Wu2021,WuImaginarityPRA2021,Wu2024}}. Conversely, up to unitary equivalence, the joint requirements that every pure state be alchemical and that Eq.~\eqref{eq:intro-bound} be saturated for every pure resource characterize aquaternionicity.
{The inverse-problem landscape and the corresponding geometry of pure alchemical states are summarized schematically in Fig.~\ref{fig:inverse-landscape}, while Table~\ref{tab:resource-theory-comparison} compares the three theories considered here.}

\begin{figure}[t]
\centering
\makebox[\columnwidth][c]{
\hspace*{-0.129cm}
\begin{tikzpicture}[
scale=0.97,
transform shape,
>=Latex,
font=\normalsize,
line cap=round,
line join=round,
box/.style={draw=black,rounded corners=2pt,fill=gray!6,minimum width=3.65cm,minimum height=0.95cm,align=center},
theory/.style={draw=black!50,rounded corners=2pt,fill=white,minimum width=2.40cm,minimum height=1.10cm,align=center},
aqtheory/.style={draw=aqorange!85,rounded corners=2pt,fill=aqorange!7,minimum width=3.30cm,minimum height=1.10cm,align=center},
smallnote/.style={font=\footnotesize,align=center}
]

\node[box,minimum width=5.77cm,minimum height=1.00cm] (inverse) at (3.80,5.45) {Alchemical resource theory\\[0.5mm]\footnotesize replication $+$ exact universal simulation};

\node[box,draw=black!50] (qubit) at (1.35,4.00) {local dimension $d=2$\\[0.5mm]\footnotesize multi-qubit system};
\node[box,minimum width=4.40cm] (qudit) at (6.00,4.00) {even local dimension $d\ge4$\\[0.5mm]\footnotesize multi-\textbf{qudit} system};

\draw[->,thick,black] (inverse.south -| qubit.north) -- (qubit.north);
\draw[->,thick,black] (inverse.south -| qudit.north) -- (qudit.north);

\node[theory] (im) at (0.75,2.20) {\textbf{Imaginarity}\\[0.5mm]\footnotesize $K^2=I$};
\node[theory] (par) at (3.45,2.20) {\textbf{Parity}\\[0.5mm]\footnotesize $Z^2=I$};
\node[aqtheory] (aq) at (6.60,2.20) {\textbf{\color{aqorange}Aquaternionicity}\\[0.5mm]\footnotesize $\Theta^2=-I$};

\draw[->,black!50] (qubit.south) -- (im.north);
\draw[->,black!50] (qubit.south) -- (par.north);

\draw[->,black!50] (qudit.south) -- (im.north east);
\draw[->,black!50] (qudit.south) -- (par.north east);
\draw[->,thick,aqorange!80!black] (qudit.south) -- ([xshift=-0.22cm]aq.north);

\node[smallnote,text=red!80!black,font=\footnotesize\bfseries] at (7.3,2.95) {New theory};

\draw[gray!60,fill=gray!4] (0.75,0.95) circle (0.50);
\draw[gray!35] (0.25,0.95) arc (180:360:0.50 and 0.14);
\draw[gray!35] (0.25,0.95) arc (180:0:0.50 and 0.14);
\fill[imblue] (0.40,0.85) circle (1.8pt);
\fill[imblue] (1.10,1.05) circle (1.8pt);
\node[smallnote] at (0.7,0.295) {fine-tuned subset};
\node[smallnote] at (0.74,-0.03) {$\psi^T\psi=0$};

\draw[gray!60,fill=gray!4] (3.45,0.95) circle (0.50);
\draw[parblue,line width=1.0pt] (2.95,0.95) arc (180:360:0.50 and 0.14);
\draw[parblue,line width=1.0pt] (2.95,0.95) arc (180:0:0.50 and 0.14);
\node[smallnote] at (3.65,0.10) {codimension-1 subset\\${\langle\psi|Z|\psi\rangle=0}$};

\draw[aqorange!85,fill=aqorange!14,line width=0.9pt] (6.60,0.95) circle (0.50);
\draw[aqorange!70,line width=0.7pt] (6.10,0.95) arc (180:360:0.50 and 0.14);
\draw[aqorange!70,line width=0.7pt] (6.10,0.95) arc (180:0:0.50 and 0.14);
\draw[aqorange!70,line width=0.7pt] (6.60,0.45) arc (-90:90:0.14 and 0.50);
\node[smallnote,text=aqorange!90!black] at (6.60,0.10) {\textbf{all pure states}\\$\langle\Theta\psi|\psi\rangle=0$};

\end{tikzpicture}
}
\vspace{-0.24cm}
\caption{{Selected landscape for alchemical resource theories.} For multi-qubit systems, the inverse problem yields imaginarity and parity asymmetry. For even local dimension $d\ge4$, aquaternionicity emerges as an additional quaternionic solution, for which every pure state is alchemical.}
\label{fig:inverse-landscape}
\end{figure}
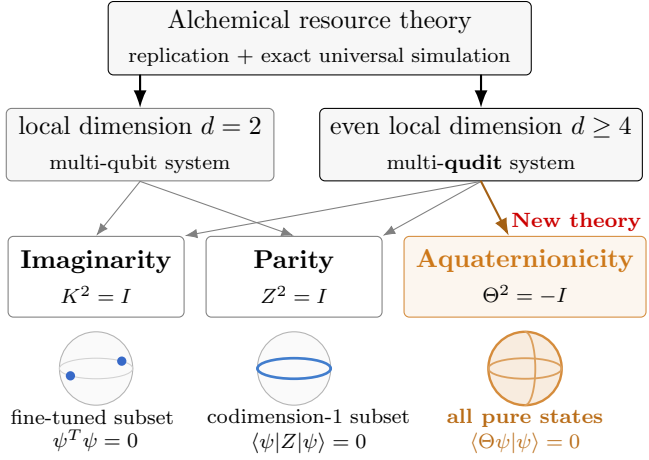

{The same support geometry also enables \emph{recalibration-free exact simulation} under structured noise: if all noisy realizations remain within a common isotropic subspace, a single free processor works exactly throughout the noise family. We demonstrate this mechanism for basis dephasing \cite{Helm2011,Yao2017} and pairwise multilevel amplitude damping \cite{LeungNielsenChuangYamamoto1997,ChessaGiovannetti2021}, where aquaternionicity admits larger noise-invariant alchemical supports than imaginarity and parity asymmetry.}

\begin{table*}[t]
\caption{
Comparison of imaginarity (Im), parity asymmetry (Par), and
aquaternionicity (Aq) for even local dimension \(d\ge4\).
Here \(K\) denotes complex conjugation in a fixed reference basis,
\(Z=I_{d/2}\oplus(-I_{d/2})\), and
\(\Theta\) is antiunitary with \(\Theta^2=-I\).
{In the \(\dim\mathcal M_g\) row, \(g\) denotes a pure alchemical state.}
}
\label{tab:resource-theory-comparison}

\centering
\footnotesize
\setlength{\tabcolsep}{3pt}
\renewcommand{\arraystretch}{1.18}

\begin{ruledtabular}
{
\begin{tabular}{c@{\hspace{0.65cm}}c@{\hspace{0.65cm}}c@{\hspace{0.65cm}}c}
Property
& Imaginarity
& \shortstack{Parity\\asymmetry}
& Aquaternionicity
\\
\hline

Free condition
& \(K\rho K^{-1}=\rho\)
& \(Z\rho Z=\rho\)
& \(\Theta\rho\Theta^{-1}=\rho\)
\\[1mm]

Pure alchemical
& \(\psi^{T}\psi=0\)
& \(\langle\psi|Z|\psi\rangle=0\)
& All pure states
\\[1mm]

\(\dim\mathcal M_g\)
& \(2d-6\)
& \(2d-5\)
& \(2d-4\)
\\[1mm]

Recalibration-free
& \shortstack{Suitably aligned\\supports}
& \shortstack{Suitably aligned\\supports}
& \shortstack{Noise-adapted pure state\\+ maximal mixing}
\\

\end{tabular}
}
\end{ruledtabular}
\end{table*}

\paragraph{Replicable resources and universal mixing bound.---}
\label{sec:framework}

Let the elementary system be a \(d\)-level qudit with Hilbert space
\(\mathcal H=\mathbb C^d\), and let
\(\mathcal H_m=\mathcal H^{\otimes m}\).
We denote by \(\mathcal D_m\) and \(\mathcal P_m\) the density operators and pure states on \(\mathcal H_m\), respectively.
{Unless explicitly specified as complex or affine, all dimensions are real.}
A resource theory is specified by nonempty closed convex sets of free states
\(\mathcal F_m\subseteq\mathcal D_m\), with
\(\mathcal F_1\neq\mathcal D_1\), satisfying the standard consistency conditions under tensor products, marginalization, and permutations.
{We call such a theory \emph{nontrivial} when \(\mathcal F_1\) is not a singleton; equivalently, its single-system free-state space is neither all of \(\mathcal D_1\) nor a one-point set.}
{Free operations are defined as the free-ancilla closure of completely free Kraus (CFK) instruments.}
Precise definitions are given in the Supplemental Material.

A nonfree state \(\omega\) is \emph{alchemical} if it is exactly replicable by free operations and enables the exact simulation of arbitrary finite-system quantum instruments using finitely many copies.
Namely, there exists a free channel
\(\Lambda_\omega^{\rm rep}\) such that
\[
\Lambda_\omega^{\rm rep}(\omega)=\omega\otimes\omega,
\]
and, for every finite system and every target instrument \(\{\Phi_x\}_x\), a finite \(k\) and a free instrument \(\{\Lambda_x\}_x\) such that
\[
\Lambda_x\!\left(\rho\otimes\omega^{\otimes k}\right)
=
\Phi_x(\rho)
\qquad
(\forall\rho,\ \forall x).
\]
Let \(\mathcal A\) denote the set of alchemical states.

We ask how broadly such resources can populate the quantum state space.
The two previously known examples of alchemical resource theories are imaginarity and parity asymmetry.
{For parity asymmetry in even dimension, we take \(Z:=I_{d/2}\oplus(-I_{d/2})\).}
In these theories, a pure state
\(g=\ket{\psi}\!\bra{\psi}\) is alchemical only when, respectively,
\[
\psi^T\psi=0,
\qquad\text{or}\qquad
\langle\psi|Z|\psi\rangle=0,
\]
Thus, only restricted subsets of the pure-state space are alchemical.
This raises the first question: can \emph{every} pure state be alchemical?

{The second question concerns mixing around a pure alchemical resource. For a pure \(g\in\mathcal A\), we define the
\emph{admissible pure-mixing set}}
\[
\mathcal M_g
:=
\left\{
h\in\mathcal P_1:
(1-t)g+th\in\mathcal A
\ \text{for every }0<t<1
\right\}.
\]
Our second question is how large \(\mathcal M_g\) can be.
We find the following universal bound.

\begin{theorem}[Universal mixing bound]
\label{thm:universal-mixing-bound}
For every nontrivial resource theory satisfying the above conditions and every pure alchemical state \(g\),
\begin{equation}
{\dim\mathcal M_g\le 2d-4.}
\label{eq:universal-mixing-bound}
\end{equation}
\end{theorem}

{For even \(d\ge4\), imaginarity and parity asymmetry do not attain the bound:}
\[
\dim\mathcal M_g^{\rm Im}=2d-6,
\qquad
\dim\mathcal M_g^{\rm Par}=2d-5.
\]
{Although the gaps are only constant in \(d\), their relative significance is largest at small \(d\).}
We next construct a resource theory that satisfies both properties: every pure state is alchemical, and the universal bound is saturated for every pure state.


\paragraph{Aquaternionicity and maximal state-space abundance.---}
\label{sec:aquaternionicity}

{Let \(d=2n\ge4\)} and fix an antiunitary operator
\(\Theta:\mathcal H\to\mathcal H\) satisfying
\begin{equation}
\Theta^2=-I .
\label{eq:kramers-antiunitary}
\end{equation}
The Hermitian operators invariant under
\(X\mapsto\Theta X\Theta^{-1}\)
form the standard complex representation of the quaternionic Hermitian algebra
\(H_n(\mathbb H)\).
Motivated by this quaternionic structure, we define
\emph{aquaternionicity} as the asymmetry resource theory whose single-system free states are
\begin{equation}
\mathcal F_1^{\rm Aq}
:=
\left\{
\rho\in\mathcal D_1:
\Theta\rho\Theta^{-1}=\rho
\right\}.
\label{eq:aq-free-states-single}
\end{equation}
{
\noindent\emph{Remark (qubit degeneration).---}
The case \(d=2\) corresponds to \(n=1\), for which the quaternionic Hermitian algebra is simply \(H_1(\mathbb H)=\mathbb R\). In its standard complex representation, the \(\Theta\)-invariant Hermitian operators are therefore only scalar multiples of \(I_2\), so normalization leaves the single free state \(\mathcal F_1^{\rm Aq}=\{I/2\}\). Thus the \(d=2\) construction is a one-point free-state theory and hence trivial in the sense used here. The first nontrivial free-state space of this quaternionic form arises at \(n=2\), corresponding to \(d=4\); accordingly, aquaternionicity is defined here for \(d=2n\ge4\).
}
For \(m\) elementary systems, we take the canonical tensor-product extension
\[
\Theta_m:=\Theta^{\otimes m},
\]
with free states
\begin{equation}
\mathcal F_m^{\rm Aq}
:=
\left\{
\rho\in\mathcal D_m:
\Theta_m\rho\Theta_m^{-1}=\rho
\right\}.
\label{eq:aq-free-states}
\end{equation}
{We define the free instruments of aquaternionicity as the branchwise \(\Theta\)-covariant instruments. As shown in the Supplemental Material, this covariance-defined class coincides with the free-ancilla closure of the CFK class used in the general framework above.}
An explicit block-matrix representation of Eq.~\eqref{eq:aq-free-states-single} is given in the End Matter.

The key geometric consequence of
\(\Theta^2=-I\)
is Kramers orthogonality
{\cite{FruchartYaoVitelli2023}},
\begin{equation}
\langle\Theta\psi|\psi\rangle=0
\qquad
(\forall\,\ket{\psi}\in\mathcal H).
\label{eq:kramers-orthogonality}
\end{equation}
We call a subspace \(S\subseteq\mathcal H\) \emph{isotropic} when it is orthogonal to \(\Theta S\), i.e.,
\[
S\perp\Theta S .
\]
In aquaternionicity, alchemicality is characterized entirely by this support geometry:
\begin{equation}
{
\omega\in\mathcal A^{\rm Aq}
\quad\Longleftrightarrow\quad
\operatorname{supp}\omega
\perp
\Theta\,\operatorname{supp}\omega .
}
\label{eq:aq-alchemical-support}
\end{equation}
In particular, Eq.~\eqref{eq:kramers-orthogonality} makes every one-dimensional support isotropic.

\begin{theorem}[Extremality of aquaternionicity]
\label{thm:aq-extremality}
For aquaternionicity on
\(\mathcal H=\mathbb C^d\), {\(d=2n\ge4\)}:

\begin{enumerate}
\item
Every pure state is alchemical.

\item
For every pure state
\(g=\ket{\psi}\!\bra{\psi}\),
\begin{equation}
\mathcal M_g^{\rm Aq}
=
\left\{
\ket{\phi}\!\bra{\phi}:
\langle\Theta\psi|\phi\rangle=0
\right\},
\label{eq:aq-admissible-set}
\end{equation}
and hence
\begin{equation}
{
\dim\mathcal M_g^{\rm Aq}=2d-4
}
\label{eq:aq-uniform-saturation}
\end{equation}
for every pure \(g\).
\end{enumerate}
\end{theorem}

Aquaternionicity also admits counterparts of the principal operational tasks established for imaginarity, including resource distillation, assisted state conversion, and resource-assisted channel discrimination
\cite{Wu2021,Wu2024};
details are given in the Supplemental Material.
The larger admissible mixing set has a simple geometric origin: the additional self-condition present in imaginarity and parity asymmetry becomes automatic in aquaternionicity by Kramers orthogonality.

We now ask whether the two extremal properties above merely hold for aquaternionicity or instead characterize its underlying structure.

\begin{theorem}[Characterization of aquaternionicity]
\label{thm:aq-characterization}
Consider a nontrivial resource theory {satisfying the framework assumptions introduced above}.
Suppose that

\begin{enumerate}
\item every pure state \(g\in\mathcal P_1\) is alchemical, and
\item the universal mixing bound is saturated for every pure state,
\begin{equation}
\dim\mathcal M_g=2d-4
\qquad
(\forall g\in\mathcal P_1).
\label{eq:uniform-saturation-assumption}
\end{equation}
\end{enumerate}

{Then \(d=2n\ge4\)} and \(\mathcal F_m=\mathcal F_m^{\rm Aq}\) for all \(m\ge1\), for an antiunitary \(\Theta\), unique up to an overall phase, satisfying \(\Theta^2=-I\).
\end{theorem}

Thus, the joint requirements that every pure state be alchemical and that the universal mixing bound be saturated determine the full aquaternionicity hierarchy. {Together with the CFK--covariance equivalence established in the Supplemental Material, this also fixes the associated free-instrument class to be the \(\Theta\)-covariant one.}


\paragraph{Recalibration-free exact simulation under noise.---}
\label{sec:noise}

We now turn the abundance of alchemical states in aquaternionicity into an operational advantage under noise.
Because every pure state is alchemical, the initial resource can be chosen solely for compatibility with a known noise structure, without any additional restriction on its direction.
More generally, Eq.~\eqref{eq:aq-alchemical-support} shows that every state supported on an isotropic subspace is alchemical.
{Crucially, this pointwise statement does not by itself imply uniformity: a priori, different resource states supported on the same subspace could require different simulation instruments. We prove the uniform statement formally in the Supplemental Material; the explicit controlled-Kraus processor realizing it is displayed in the End Matter.}

For a subspace \(S\subseteq\mathcal H\), let
\[
\mathcal D(S)
:=
\left\{
\omega\in\mathcal D_1:
\operatorname{supp}\omega\subseteq S
\right\}.
\]
If \(S\perp\Theta S\), then for every finite-system target instrument
\(\{\Phi_x:m\to m\}_x\), there exists a single free instrument
\(\{\Lambda_{x,S}\}_x\) satisfying
\begin{equation}
{
\Lambda_{x,S}(\rho\otimes\omega)
=
\Phi_x(\rho)\otimes\omega
}
\label{eq:uniform-support-simulation}
\end{equation}
for all
\(\rho\in\mathcal D_m\),
\(\omega\in\mathcal D(S)\), and \(x\).
The simulation map may depend on \(S\) and on the target instrument, but not on the particular resource state \(\omega\) within \(S\).
Moreover, the resource is returned unchanged and can therefore be reused.

This provides a simple mechanism for exact simulation under uncertain noise.
Consider a family of noise processes
\(\{\Gamma_\alpha\}_\alpha\)
whose structure is known while the realized parameter \(\alpha\) may vary.
If a pure resource \(g\) can be chosen such that
\begin{equation}
\operatorname{supp}\Gamma_\alpha(g)\subseteq S
\qquad
(\forall\alpha)
\label{eq:noise-family-support}
\end{equation}
for a common isotropic subspace \(S\), then the same instrument
\(\{\Lambda_{x,S}\}_x\) remains valid for every noisy realization.
The resource may become mixed or lose coherence within \(S\), yet neither tomography of the realized state nor retuning of the simulation map is required.
We refer to this as \emph{recalibration-free exact simulation}.

The mixing geometry introduced above is the simplest instance of this support principle.
For
\(g=\ket{\psi}\!\bra{\psi}\) and
\(h=\ket{\phi}\!\bra{\phi}\in\mathcal M_g^{\rm Aq}\),
the two-dimensional support
\(\operatorname{span}_{\mathbb C}\{\psi,\phi\}\)
is isotropic.
Hence all mixtures
\((1-t)g+th\)
can be used with the same exact simulation instrument.
Thus the saturation of \(\dim\mathcal M_g\) directly enlarges the set of mixing directions compatible with recalibration-free simulation.

\paragraph{Standard noise models.---}
We illustrate this mechanism with two standard noise structures.
{Let \(d=2n\ge4\)}, and choose an orthonormal basis
\(\{e_j,f_j\}_{j=1}^n\).
Let \(K\) be complex conjugation in this basis, i.e., \(Ke_j=e_j\) and \(Kf_j=f_j\), and choose the basis such that
\[
Ze_j=e_j,\quad Zf_j=-f_j,\qquad
\Theta e_j=f_j,\quad \Theta f_j=-e_j.
\]

For the three theories, a support on which every state is alchemical satisfies, respectively,
\[
S\perp KS,\qquad
S\perp ZS,\qquad
S\perp\Theta S.
\]

First consider basis dephasing,
\begin{equation}
\begin{aligned}
\Gamma_{\lambda}^{\rm deph}(X)
={}&
\lambda X
+(1-\lambda)\sum_{j=1}^{n}
\Bigl[
\proj{e_j}X\proj{e_j}\\
&\hspace{22mm}
+\proj{f_j}X\proj{f_j}
\Bigr].
\end{aligned}
\label{eq:basis-dephasing}
\end{equation}
with \(0\le\lambda<1\).
For aquaternionicity,
\[
E:=\operatorname{span}_{\mathbb C}\{e_1,\ldots,e_n\}
\]
is both noise invariant and isotropic.
Hence every state initially supported on \(E\) remains within
\(\mathcal D(E)\), even when coherence is lost and a pure state becomes mixed, and the same exact simulation instrument remains valid throughout the process.
In contrast, no nonzero coordinate subspace satisfies the corresponding imaginarity or parity-asymmetry condition.

As a dissipative example, consider pairwise multilevel amplitude damping
{\cite{LeungNielsenChuangYamamoto1997,ChessaGiovannetti2021}},
with
\begin{equation}
\begin{aligned}
A_0&=P_E+\sqrt{1-\gamma}\,P_F,\\
A_j&=\sqrt{\gamma}\,\ket{e_j}\!\bra{f_j},
\qquad j=1,\ldots,n .
\end{aligned}
\label{eq:pair-amplitude-damping}
\end{equation}
where
\(E=\operatorname{span}_{\mathbb C}\{e_j\}_{j=1}^n\),
\(F=\operatorname{span}_{\mathbb C}\{f_j\}_{j=1}^n\),
\(P_E\) and \(P_F\) are the orthogonal projectors onto \(E\) and \(F\), respectively,
and \(0<\gamma\le1\).
These Kraus operators describe the relaxation
\(f_j\rightarrow e_j\).
For aquaternionicity, \(E\) again provides an \(n\)-dimensional noise-invariant alchemical support.
For parity asymmetry no nonzero support has both properties, whereas for imaginarity the maximal dimension is
\(\lfloor n/2\rfloor\).
{Detailed derivations are given in the Supplemental Material; the corresponding noise-model calculations are summarized in the End Matter.}

The separation is already visible for \(d=4\).
Table~\ref{tab:d4-noise-comparison} gives the maximal complex dimension of a noise-invariant alchemical support on which one exact simulation instrument remains valid.

\begin{table}[t]
\caption{
Maximal complex dimension of a noise-invariant alchemical support
admitting a single exact simulation instrument for \(d=4\).
}
\label{tab:d4-noise-comparison}
\begin{ruledtabular}
\begin{tabular}{lccc}
Noise
& \(\mathrm{Im}\)
& \(\mathrm{Par}\)
& \(\mathrm{Aq}\)
\\
\hline
Basis dephasing
& \(0\)
& \(0\)
& \(2\)
\\
Pair amplitude damping
& \(1\)
& \(0\)
& \(2\)
\end{tabular}
\end{ruledtabular}
\end{table}

{These examples do not imply a universal ordering under arbitrary noise.}

\paragraph{Conclusion.---}

We introduced aquaternionicity, an alchemical resource theory associated with a Kramers antiunitary satisfying \(\Theta^2=-I\), and showed that it is extremal in the abundance of resources for exact universal simulation.  Every pure state is alchemical, and for every pure resource the admissible mixing set saturates the universal bound \(\dim\mathcal M_g=2d-4.\)
Conversely, these two extremal properties characterize the full composite aquaternionicity hierarchy.  The same geometry also yields recalibration-free exact simulation: when noisy realizations remain within a common isotropic support, a single free simulation map works exactly without identifying the realized resource state or retuning the protocol.

Several questions remain open.  First, it is important to identify concrete physical realizations of aquaternionicity and to determine how its resource-theoretic advantages survive realistic imperfections.  Kramers pairs are themselves physical objects, arising naturally in systems with antiunitary symmetries satisfying \(\Theta^2=-I\){~\cite{HasanKane2010}}, suggesting that the underlying structure need not be merely formal; the open problem is to turn this structure into an experimentally controllable resource theory with suitable state preparation, free operations, and noise models.  {Second, it remains unknown whether the extremal degree of recalibration robustness identified here is specific to the quaternionic/Kramers structure.}  More generally, one may ask whether other alchemical resource theories exist whose resource states occupy comparably large regions of state space and support exact simulation uniformly over nontrivial noisy families.  Classifying such theories would extend the inverse problem from identifying replicable universal resources to identifying resource structures that are intrinsically robust to uncertainty in their physical realization.

\paragraph{Acknowledgments.---}
The authors used OpenAI's ChatGPT (GPT-5.5, GPT-5.6, and GPT-6) to assist with discussions, verification of mathematical proofs, and drafting and revising the manuscript. The concept of aquaternionicity and its underlying scientific idea were discovered by the authors without AI assistance. The AI was used under detailed and iterative instructions specifying the mathematical arguments, logical structure, and intended presentation. All AI-assisted outputs, including mathematical arguments and manuscript text, were thoroughly and independently checked, revised, and approved by the authors; the final proofs, scientific claims, and presentation were determined by the authors.
H.A.~was supported by JSPS KAKENHI Grant Nos.~25KJ0043, 26K17031, and 26K21941.


\clearpage
\onecolumngrid
\section*{End Matter}

\paragraph{Explicit quaternionic realization.---}
The single-system structure of aquaternionicity can be represented explicitly by choosing
\(\Theta=J_nK\), where
\(J_n=\left(\begin{smallmatrix}0&-I_n\\ I_n&0\end{smallmatrix}\right)\)
and \(K\) denotes complex conjugation.
Then \(\Theta^2=-I\), and a Hermitian operator is invariant under
\(X\mapsto\Theta X\Theta^{-1}\) precisely when
\[
X=
\begin{pmatrix}
A&B\\
-\overline B&\overline A
\end{pmatrix},
\qquad
A=A^\dagger,\quad B=-B^T .
\]
This is the standard complex representation of the quaternionic Hermitian algebra
\(H_n(\mathbb H)\) \cite{Baez2012}.
The bilinear form
\(\Omega(x,y):=\langle\Theta x|y\rangle\)
is alternating, so \(\Omega(x,x)=0\) for every \(x\).
Consequently, an isotropic subspace \(S\perp\Theta S\) has complex dimension at most \(n=d/2\).
If equality holds, then
\(\mathcal H=S\oplus\Theta S\).

For \(d=4\) (\(n=2\)), choose the ordered basis
\(\{e_1,e_2,f_1,f_2\}\) such that
\(\Theta e_j=f_j\) and
\(\Theta f_j=-e_j\).
Then
\[
J_2=
\begin{pmatrix}
0&0&-1&0\\
0&0&0&-1\\
1&0&0&0\\
0&1&0&0
\end{pmatrix}.
\]
Writing
\(A=\left(\begin{smallmatrix}a&z\\ \overline z&b\end{smallmatrix}\right)\)
and
\(B=\left(\begin{smallmatrix}0&w\\-w&0\end{smallmatrix}\right)\),
an aquaternionicity-free state has the explicit form
\[
\rho=
\begin{pmatrix}
a&z&0&w\\
\overline z&b&-w&0\\
0&-\overline w&a&\overline z\\
\overline w&0&z&b
\end{pmatrix},
\]
where \(a,b\in\mathbb R\), \(z,w\in\mathbb C\),
\(\rho\ge0\), and \(2(a+b)=1\).
Before normalization these matrices have six real parameters, reproducing
\(\dim_{\mathbb R}H_2(\mathbb H)=6\); the trace-one free-state set therefore has affine dimension \(5\).
{The excluded \(d=2\) one-point degeneration is discussed immediately after the definition of aquaternionicity in the main text.}

\paragraph{Mixing geometry.---}
The different values of \(\dim\mathcal M_g\) in imaginarity, parity asymmetry, and aquaternionicity follow directly from their support conditions.
Let \(d=2n\ge4\),
let \(g=\ket{\psi}\!\bra{\psi}\) be pure and alchemical, and let
\(h=\ket{\phi}\!\bra{\phi}\).
For \(0<t<1\), the support of
\((1-t)g+th\) is
\(\operatorname{span}_{\mathbb C}\{\psi,\phi\}\).

For imaginarity, alchemicality requires
\(\psi^T\phi=0\) and \(\phi^T\phi=0\), while
\(\psi^T\psi=0\) already holds for \(g\).
The first is one complex linear condition and the second one further complex condition, giving
\(\dim\mathcal M_g^{\rm Im}=2d-6\).

For parity asymmetry with
\(Z=I_n\oplus(-I_n)\), the corresponding conditions are
\(\langle\psi|Z|\phi\rangle=0\) and
\(\langle\phi|Z|\phi\rangle=0\).
They impose two and one real conditions, respectively, yielding
\(\dim\mathcal M_g^{\rm Par}=2d-5\).

For aquaternionicity, the self-condition
\(\langle\Theta\phi|\phi\rangle=0\)
holds automatically by Kramers orthogonality.
The only remaining condition is
\(\langle\Theta\psi|\phi\rangle=0\), and hence
\[
\mathcal M_g^{\rm Aq}
=
\left\{\ket{\phi}\!\bra{\phi}:\langle\Theta\psi|\phi\rangle=0\right\},
\qquad
\dim\mathcal M_g^{\rm Aq}=2d-4 .
\]

{
\paragraph{Uniform simulation on an isotropic support.---}
We record here the explicit processor underlying
Eq.~\eqref{eq:uniform-support-simulation}; its freeness, normalization,
and uniform action on the full support are proved in the Supplemental Material.
Extend an isotropic subspace \(S\) to an \(n\)-dimensional maximal isotropic
subspace \(L\), and let \(P=P_L\) and
\(\overline P=P_{\Theta L}\), so that \(P+\overline P=I\).
For a target instrument
\(\Phi_x(X)=\sum_\mu A_{x\mu}XA_{x\mu}^\dagger\)
on \(m\) systems, define
\(A_{x\mu}^{\#}:=\Theta_mA_{x\mu}\Theta_m^{-1}\) and
\begin{equation}
K_{x\mu}
=
A_{x\mu}\otimes P
+
A_{x\mu}^{\#}\otimes\overline P .
\label{eq:em-controlled-kraus}
\end{equation}
With
\(\Lambda_{x,L}(X):=\sum_\mu K_{x\mu}XK_{x\mu}^\dagger\),
this single processor satisfies
\[
\Lambda_{x,L}(\rho\otimes\omega)
=
\Phi_x(\rho)\otimes\omega
\qquad
\bigl(\omega\in\mathcal D(L)\bigr),
\]
and hence applies in particular to every \(\omega\in\mathcal D(S)\).
}

\paragraph{Operational protocol for recalibration-free simulation.---}
The uniform construction above yields a two-stage protocol.
Suppose that the structure of a noise family
\(\{\Gamma_\alpha\}_\alpha\)
is known, while the realized parameter \(\alpha\) is not.

In the \emph{design stage}, choose a pure alchemical resource \(g\) and an isotropic subspace \(S\) such that
\(\operatorname{supp}\Gamma_\alpha(g)\subseteq S\)
for every allowed \(\alpha\).
Extend \(S\), if necessary, to an \(n\)-dimensional isotropic subspace \(L\), and compile
\(\{\Lambda_{x,L}\}_x\)
for the desired target instrument using
Eq.~\eqref{eq:em-controlled-kraus}.
This stage may depend on the known noise structure, but not on the subsequently realized value of \(\alpha\).

In the \emph{operation stage}, prepare \(g\), let the unknown noise realization act, and {feed the resulting state
\(\Gamma_\alpha(g)\)
directly into the precompiled simulator.}
Since
\(\Gamma_\alpha(g)\in\mathcal D(L)\),
\[
\Lambda_{x,L}
\bigl(\rho\otimes\Gamma_\alpha(g)\bigr)
=
\Phi_x(\rho)\otimes\Gamma_\alpha(g)
\]
for every allowed \(\alpha\).
After the initial design, neither estimation of the noise parameter nor tomography of the realized resource state is required.
The resource may change its purity and coherence while the same exact simulation instrument remains valid.
This separation between the design and operation stages is what we mean by \emph{recalibration-free}.

\paragraph{Four-level dephasing example.---}
The protocol is particularly transparent for \(d=4\).
Using the Kramers pairs above, let
\(E=\operatorname{span}_{\mathbb C}\{e_1,e_2\}\).
Then \(E\perp\Theta E\), and
\(\dim_{\mathbb C}E=2=d/2\).

Choose
\(\ket{\psi}=(\ket{e_1}+\ket{e_2})/\sqrt2\)
and
\(g=\ket{\psi}\!\bra{\psi}\).
Under the basis-dephasing channel of
Eq.~\eqref{eq:basis-dephasing}, including the noiseless endpoint
\(\lambda=1\), the resource becomes
\[
\omega_\lambda
=
\frac12
\left(
\proj{e_1}
+
\proj{e_2}
+
\lambda\ket{e_1}\!\bra{e_2}
+
\lambda\ket{e_2}\!\bra{e_1}
\right).
\]
As \(\lambda\) decreases, the resource loses coherence and becomes mixed, but
\(\operatorname{supp}\omega_\lambda\subseteq E\)
for every \(\lambda\).

For a target instrument
\(\Phi_x(X)=\sum_\mu A_{x\mu}XA_{x\mu}^\dagger\),
compile once the processor with
\(K_{x\mu}
=
A_{x\mu}\otimes P_E+
A_{x\mu}^{\#}\otimes P_{\Theta E}\).
It then satisfies
\[
\Lambda_{x,E}(\rho\otimes\omega_\lambda)
=
\Phi_x(\rho)\otimes\omega_\lambda
\]
for every \(0\le\lambda\le1\).
Hence the same processor remains exact throughout the dephasing process without learning \(\lambda\) or recalibrating the simulation map.

\paragraph{Standard-noise calculations.---}
We finally summarize the calculations underlying
Table~\ref{tab:d4-noise-comparison}, using the basis
\(\{e_j,f_j\}_{j=1}^n\)
{introduced in the noise analysis of the main text}.

For basis dephasing with \(0\le\lambda<1\), every invariant support is a coordinate subspace.
In imaginarity, any nonzero coordinate subspace contains a real basis vector and therefore cannot satisfy
\(S\perp KS\).
In parity asymmetry, a coordinate basis vector is a \(Z\)-eigenvector and hence cannot belong to a nonzero subspace satisfying
\(S\perp ZS\).
For aquaternionicity,
\(E=\operatorname{span}_{\mathbb C}\{e_1,\ldots,e_n\}\)
is both invariant and isotropic.
Thus the maximal complex dimensions are
\[
\dim_{\mathbb C}S_{\max}^{\rm Im}=0,\qquad
\dim_{\mathbb C}S_{\max}^{\rm Par}=0,\qquad
\dim_{\mathbb C}S_{\max}^{\rm Aq}=n .
\]

For pairwise amplitude damping, any nonzero component in
\(F=\operatorname{span}_{\mathbb C}\{f_1,\ldots,f_n\}\)
generates the corresponding \(e_j\) through the jump operators.
For aquaternionicity this is incompatible with isotropy, so an invariant alchemical support must lie in \(E\); the full \(E\) is admissible and has dimension \(n\).
For parity asymmetry, no nonzero invariant support satisfies the alchemical condition.
For imaginarity, an invariant alchemical support must again lie in \(E\), where it must be totally isotropic for the symmetric bilinear form \(x^Ty\).
Its maximal dimension is \(\lfloor n/2\rfloor\).
Therefore
\begin{equation}
\dim_{\mathbb C}S_{\max}^{\rm Im}
=
\left\lfloor\frac n2\right\rfloor,
\qquad
\dim_{\mathbb C}S_{\max}^{\rm Par}=0,
\qquad
\dim_{\mathbb C}S_{\max}^{\rm Aq}=n .
\label{eq:em-damping-dimensions}
\end{equation}
For \(d=4\) (\(n=2\)), these results give
\((0,0,2)\) for basis dephasing and
\((1,0,2)\) for pairwise amplitude damping, as reported in
Table~\ref{tab:d4-noise-comparison}.


\clearpage
\onecolumngrid

\setcounter{section}{0}
\setcounter{equation}{0}
\setcounter{assumption}{0}
\setcounter{capability}{0}
\setcounter{proposition}{0}
\setcounter{lemma}{0}
\setcounter{theorem}{0}
\setcounter{corollary}{0}
\setcounter{definition}{0}

\newcommand{\Herm}{\operatorname{Herm}}
\newcommand{\spanR}{\operatorname{span}_{\mathbb R}}
\newcommand{\id}{\operatorname{id}}

\begin{center}
{\large\bfseries Supplemental Material for ``Resource Theory of Aquaternionicity: Every Pure State Is a Replicable Resource for Exact Simulation of Arbitrary Quantum Operations''}
\end{center}

\input{PRL_supplement_body_v1.tex}

\end{document}

%% file: PRL_supplement_body_v1.tex
%

\setcounter{section}{0}
\setcounter{secnumdepth}{1}
\renewcommand{\thesection}{S\arabic{section}}
\renewcommand{\theequation}{S\arabic{equation}}
\renewcommand{\theassumption}{S\arabic{assumption}}
\renewcommand{\thecapability}{S\arabic{capability}}
\renewcommand{\theproposition}{S\arabic{proposition}}
\renewcommand{\thelemma}{S\arabic{lemma}}
\renewcommand{\thetheorem}{S\arabic{theorem}}
\renewcommand{\thecorollary}{S\arabic{corollary}}
\renewcommand{\thedefinition}{S\arabic{definition}}

\section{General resource-theoretic framework}
\label{sec:sm-framework}

We first state the general framework used throughout this Supplemental Material.
We follow the alchemical resource theory introduced in Ref.~\cite{InverseDraft},
with the free-operation class closed under the addition of free ancillary
systems as described below.  Unless explicitly specified as complex or
affine, all dimensions below are real.

\subsection{Free states and completely free Kraus operations}

Let the elementary system have Hilbert space
\(\mathcal H=\mathbb C^d\), and let
\(\mathcal H_m=\mathcal H^{\otimes m}\).
We denote by \(\mathcal D_m\) and \(\mathcal P_m\) the sets of density
operators and pure states on \(\mathcal H_m\), respectively.

For each \(m\ge1\), let
\(\mathcal F_m\subseteq\mathcal D_m\)
be the set of free states.  We assume:

\begin{enumerate}
\item[(F1)]
\(\mathcal F_m\) is nonempty, closed, and convex for every \(m\), and
\(\mathcal F_1\neq\mathcal D_1\).

\item[(F2)]
If \(\rho\in\mathcal F_m\) and
\(\sigma\in\mathcal F_n\), then
\(\rho\otimes\sigma\in\mathcal F_{m+n}\).

\item[(F3)]
Every marginal of a free state is free.

\item[(F4)]
The hierarchy \(\{\mathcal F_m\}_m\) is invariant under permutations
of the elementary systems.
\end{enumerate}

{Under (F1), we call the theory \emph{nontrivial} when \(\mathcal F_1\) is not a singleton; equivalently, the single-system free-state space is neither all of \(\mathcal D_1\) nor a one-point set.}

Define the free cone by
\[
\mathcal C_m
:=
\{a\rho:a\ge0,\ \rho\in\mathcal F_m\}.
\]
{
A linear operator
\(A:\mathcal H_m\to\mathcal H_n\)
is a \emph{completely free Kraus} (CFK) operator if
\begin{equation}
(A\otimes I_r)\,
\mathcal C_{m+r}\,
(A^\dagger\otimes I_r)
\subseteq
\mathcal C_{n+r}
\qquad
(\forall r\ge0),
\label{eq:sm-cfk}
\end{equation}
i.e., the branch generated by \(A\) preserves the free cone even in
the presence of an arbitrary reference system of \(r\) elementary
subsystems.

A completely positive (CP) map
\(\mathcal E:m\to n\) is called \emph{CFK} if it admits a Kraus
representation
\begin{equation}
\mathcal E(X)=\sum_\alpha A_\alpha X A_\alpha^\dagger
\label{eq:sm-cfk-map}
\end{equation}
in which every \(A_\alpha\) is a CFK operator.  A trace-preserving CFK
map is called a \emph{CFK channel}.

A quantum instrument
\(\{\mathcal E_x:m\to n\}_{x\in\mathsf X}\) is a finite collection of
CP trace-nonincreasing maps such that
\(\sum_{x\in\mathsf X}\mathcal E_x\) is a channel.  For an input state
\(\rho\), outcome \(x\) occurs with probability
\(p_x=\operatorname{Tr}\mathcal E_x(\rho)\), and, when \(p_x>0\), the
conditional output state is \(\mathcal E_x(\rho)/p_x\).  We call the
instrument a \emph{CFK instrument} if every branch \(\mathcal E_x\) is
a CFK CP map, i.e., if every branch admits a Kraus representation
consisting only of CFK operators.
}

The CFK condition is stable under sequential composition and tensor
products.  In particular, the identity and permutations of elementary
systems are CFK under the assumptions above.

\subsection{Free-ancilla closure}

We allow free ancillary systems to be supplied at no resource cost.
An instrument
\(\{\Lambda_x:m\to n\}_x\)
is called \emph{free} if there exist an integer \(a\ge0\), a free
ancillary state \(f\in\mathcal F_a\), and a CFK instrument
\(\{\widetilde\Lambda_x:m+a\to n\}_x\)
such that
\begin{equation}
\Lambda_x(X)
=
\widetilde\Lambda_x(X\otimes f)
\qquad
(\forall X).
\label{eq:sm-free-ancilla}
\end{equation}
We denote the resulting class of free instruments by
\(\mathcal O(m\to n)\).

{
The elementary closure properties of the free-operation class are the
same as in Ref.~\cite{InverseDraft} and will be used without further
proof:
\begin{itemize}
\item
\(\mathcal O\) is closed under sequential composition and tensor
products, and every branch
\(\Lambda_x\in\mathcal O(m\to n)\) is completely free-state
preserving, i.e.,
\[
(\Lambda_x\otimes\operatorname{id}_r)(\mathcal C_{m+r})
\subseteq
\mathcal C_{n+r}
\qquad(\forall r\ge0).
\]
\end{itemize}
}
Free-state preparation is contained in \(\mathcal O\) by supplying the
state as a free ancilla.  We do not otherwise assume that discarding a
subsystem is a free operation for an arbitrary hierarchy; whenever
reduced states are used below, their freeness follows from the
state-space assumption (F3).

\subsection{Alchemical states}

We use the updated notion of alchemicality of Ref.~\cite{InverseDraft},
in which universal simulation is required on arbitrary finite systems.

\begin{definition}[Alchemical state]
\label{def:sm-alchemical}
A state
\(\omega\in\mathcal D_1\setminus\mathcal F_1\)
is \emph{alchemical} if it satisfies:

\begin{enumerate}
\item[(A1)]
\emph{Exact replication.}
There exists a free channel
\(\mathcal R_\omega\in\mathcal O(1\to2)\) such that
\[
\mathcal R_\omega(\omega)=\omega^{\otimes2}.
\]

\item[(A2)]
\emph{Exact universal simulation.}
For every finite \(m\ge1\) and every quantum instrument
\(\{\mathcal N_x:m\to m\}_x\), there exist a finite
\(k\ge1\) and a free instrument
\(\{\Phi_x:m+k\to m\}_x\) such that
\begin{equation}
\Phi_x
\bigl(
\rho\otimes\omega^{\otimes k}
\bigr)
=
\mathcal N_x(\rho)
\qquad
(\forall\rho\in\mathcal D_m,\ \forall x).
\label{eq:sm-alchemical-simulation}
\end{equation}
\end{enumerate}

We write
\[
\mathcal A
:=
\{\omega\in\mathcal D_1:
\omega\text{ is alchemical}\}.
\]
\end{definition}

{
We will repeatedly use the following elementary consequences of
{replication and universal simulation, all of which are established in}
Ref.~\cite{InverseDraft}:
\begin{itemize}
\item
{Finite-copy simulation reduces to effective one-copy simulation:}
for every \(\omega\in\mathcal A\), every finite \(m\), and every target
instrument \(\{\mathcal N_x:m\to m\}_x\), there exists a free
instrument \(\{\widetilde\Phi_x:m+1\to m\}_x\) such that
\[
\widetilde\Phi_x(\rho\otimes\omega)=\mathcal N_x(\rho).
\]

\item
Every alchemical state is state universal.  In particular, for every
\(\omega\in\mathcal A\) and every \(\sigma\in\mathcal D_1\), there
exists a free channel \(\Lambda_{\omega\to\sigma}\) with
\[
\Lambda_{\omega\to\sigma}(\omega)=\sigma.
\]
More generally, an arbitrary finite-system state can be prepared by
choosing the corresponding replacement channel as the target.

\item
Alchemicality is transferred by free interconversion: if
\(p\in\mathcal A\) and free channels satisfy
\(C(\omega)=p\) and \(D(p)=\omega\), then \(\omega\in\mathcal A\).

\item
Consequently, whenever at least one alchemical state exists,
a single-system state is alchemical if and only if it can be freely
converted into every single-system state.
\end{itemize}
}

\subsection{Admissible mixing sets}

{
Let \(g\in\mathcal A\cap\mathcal P_1\) be a pure alchemical state.
For a fixed mixing strength \(0<t<1\), define
\[
\widehat{\mathcal M}_{g,t}
:=
\left\{
\sigma\in\mathcal D_1:
(1-t)g+t\sigma\in\mathcal A
\right\},
\qquad
\mathcal M_{g,t}
:=
\widehat{\mathcal M}_{g,t}\cap\mathcal P_1 .
\]
Thus \(\widehat{\mathcal M}_{g,t}\) allows an arbitrary mixed
second component, while \(\mathcal M_{g,t}\) is its pure-state
restriction.

The corresponding all-strength sets are
\begin{equation}
\widehat{\mathcal M}_g
:=
\bigcap_{0<t<1}\widehat{\mathcal M}_{g,t},
\qquad
\mathcal M_g
:=
\bigcap_{0<t<1}\mathcal M_{g,t}
=
\widehat{\mathcal M}_g\cap\mathcal P_1 .
\label{eq:sm-mixing-sets}
\end{equation}
Equivalently,
\[
\mathcal M_g
=
\left\{
h\in\mathcal P_1:
(1-t)g+th\in\mathcal A
\quad\forall\,0<t<1
\right\},
\]
which is the admissible pure-mixing set used in the main text.
}

\section{Universal bound on admissible mixing}
\label{sec:sm-general-bound}

We now prove the universal mixing bound stated in the main text.
The argument uses only the single-system free-state set and the fact
that free channels preserve free states; no symmetry, antiunitary
structure, or Jordan-algebraic assumption is required.

\subsection{Support constraint from a pure output}

\begin{lemma}[Effect associated with a pure output]
\label{lem:sm-effect-support}
Let \(\Lambda\) be a channel and suppose that
\(\Lambda(\rho)=g\), where \(g=\ket{\psi}\!\bra{\psi}\) is pure.
Writing \(S=\operatorname{supp}\rho\) and
\(E=\Lambda^\dagger(g)\), one has
\[
0\le E\le I,
\qquad
E\ge P_S ,
\]
where \(P_S\) denotes the projector onto \(S\).
\end{lemma}

\begin{proof}
Since
\[
\operatorname{Tr}[(I-E)\rho]
=
1-\operatorname{Tr}[g\Lambda(\rho)]
=
0
\]
and \(I-E\ge0\), we have
\((I-E)^{1/2}\rho^{1/2}=0\).
Because \(\rho\) is strictly positive on its support,
\((I-E)P_S=0\).
Hermiticity then gives
\(P_SE(I-P_S)=0\), and hence
\[
E
=
P_S+(I-P_S)E(I-P_S)
\ge P_S .
\]
\end{proof}

{
\begin{theorem}[Universal mixing bound]
\label{thm:sm-universal-mixing-bound}
For every pure alchemical state \(g\in\mathcal P_1\),
\begin{equation}
\dim\mathcal M_g\le2d-4 .
\label{eq:sm-universal-mixing-bound}
\end{equation}
\end{theorem}

\begin{proof}
We prove a stronger fixed-\(t\) statement from which
Eq.~\eqref{eq:sm-universal-mixing-bound} follows immediately.

Fix
\[
g=\ket{\psi}\!\bra{\psi}\in\mathcal A\cap\mathcal P_1
\]
and \(0<t<1\).
Consider the maximal overlap of \(g\) with the free-state set,
\[
s_g
:=
\max_{\omega\in\mathcal F_1}\operatorname{Tr}(g\omega).
\]
The maximum is attained because \(\mathcal F_1\) is compact.
Choose any optimizer
\[
{
\sigma_*
\in
\operatorname*{arg\,max}_{\sigma\in\mathcal F_1}
\operatorname{Tr}(g\sigma).
}
\]
Since \(g\) is pure and nonfree, \(s_g<1\).

Let
\[
Q_g:=I-g,
\qquad
C_g
:=
\left.
Q_g{\sigma_*}Q_g
\right|_{\psi^\perp}.
\]
The operator \(C_g\) is positive and acts on the
\((d-1)\)-dimensional space \(\psi^\perp\).
We claim that every
\(\sigma\in\widehat{\mathcal M}_{g,t}\)
is supported on the common subspace
\begin{equation}
V_g
:=
\mathbb C\psi\oplus\ker C_g,
\qquad
\operatorname{supp}\sigma\subseteq V_g
\quad
(\forall\sigma\in\widehat{\mathcal M}_{g,t}).
\label{eq:sm-common-support}
\end{equation}

To prove the claim, take
\(\sigma\in\widehat{\mathcal M}_{g,t}\) and set
\[
\rho=(1-t)g+t\sigma.
\]
By definition, \(\rho\) is alchemical.
By state universality, there exists a free channel
\(\Lambda\) satisfying
\[
\Lambda(\rho)=g.
\]
{Because both coefficients in \(\rho\) are strictly positive, let
\[
R_\rho:=\operatorname{supp}\rho
=
\mathbb C\psi+\operatorname{supp}\sigma .
\]}
Set \(E=\Lambda^\dagger(g)\).
Lemma~\ref{lem:sm-effect-support} gives
\[
{E\ge P_{R_\rho}.}
\]
Since \({\sigma_*}\) is free and \(\Lambda\) preserves free states,
\(\Lambda({\sigma_*})\) is free. Therefore, by the definition of
\(s_g\),
\[
s_g
\ge
\operatorname{Tr}[g\Lambda({\sigma_*})]
=
\operatorname{Tr}(E{\sigma_*})
\ge
{\operatorname{Tr}(P_{R_\rho}{\sigma_*}).}
\]

{Now let
\[
U_\rho:=R_\rho\cap\psi^\perp.
\]
Since \(\mathbb C\psi\subseteq R_\rho\), we have the orthogonal decomposition
\(R_\rho=\mathbb C\psi\oplus U_\rho\), and hence
\[
P_{R_\rho}=g+P_{U_\rho}.
\]}
Using
\(\operatorname{Tr}(g{\sigma_*})=s_g\),
the preceding inequality becomes
\[
s_g
\ge
s_g+{\operatorname{Tr}(P_{U_\rho}{\sigma_*}).}
\]
The last term is nonnegative, so
\[
{\operatorname{Tr}(P_{U_\rho}{\sigma_*})=0.}
\]
Positivity of \({\sigma_*}\) then implies
\[
{U_\rho\subseteq\ker{\sigma_*}.}
\]
In particular,
{\(U_\rho\subseteq\ker C_g\).}
Since
{\(\operatorname{supp}\sigma\subseteq R_\rho\), every vector in
\(\operatorname{supp}\sigma\) decomposes into a component along
\(\psi\) and a component in \(U_\rho\). Hence
\[
\operatorname{supp}\sigma
\subseteq
\mathbb C\psi\oplus U_\rho
\subseteq
\mathbb C\psi\oplus\ker C_g
=
V_g,
\]}
which proves Eq.~\eqref{eq:sm-common-support}.

It remains to bound the dimension of \(V_g\).
Because
\[
\operatorname{Tr}C_g
=
1-s_g
>
0,
\]
we have
\(\operatorname{rank}C_g\ge1\).
Therefore
\[
\dim_{\mathbb C}V_g
=
1+\dim_{\mathbb C}\ker C_g
=
d-\operatorname{rank}C_g
\le d-1.
\]
Consequently,
\begin{equation}
\dim\mathcal M_{g,t}
\le
2\dim_{\mathbb C}V_g-2
=
2(d-\operatorname{rank}C_g)-2
\le
2d-4,
\label{eq:sm-fixed-t-bound}
\end{equation}
because the pure states supported on an \(r\)-dimensional complex
subspace form \(\mathbb{CP}^{r-1}\), of real dimension \(2r-2\).
Similarly,
\[
\dim_{\rm aff}\widehat{\mathcal M}_{g,t}
\le
(\dim_{\mathbb C}V_g)^2-1
\le
(d-1)^2-1.
\]

Finally,
\[
\mathcal M_g
=
\bigcap_{0<t<1}\mathcal M_{g,t}
\subseteq
\mathcal M_{g,t}
\]
for every fixed \(t\in(0,1)\).
Combining this inclusion with
Eq.~\eqref{eq:sm-fixed-t-bound} proves
Eq.~\eqref{eq:sm-universal-mixing-bound}.
\end{proof}

The proof gives the stronger geometric statement
\eqref{eq:sm-common-support}: for every optimizer \({\sigma_*}\) of
the free overlap with \(g\), all fixed-\(t\) admissible mixing states
are confined to a proper complex subspace \(V_g\subsetneq\mathcal H\).
}

\subsection{Independence of the mixing strength}

{
We will use the following elementary observation in the proposition
below.  If a channel \(C\) maps a state \(\omega\) to a pure state
\(p\), then it maps every state supported on
\(\operatorname{supp}\omega\) to the same \(p\).  Indeed, for
\(E=C^\dagger(p)\), the argument of
Lemma~\ref{lem:sm-effect-support} gives
\(E=I\) on \(\operatorname{supp}\omega\), and hence
\(\operatorname{Tr}[pC(\eta)]=1\) for every
\(\eta\) supported there.
}

\begin{proposition}[Mixing-strength independence]
\label{prop:sm-mixing-rate-independence}
Let \(g\in\mathcal A\cap\mathcal P_1\) and
\(\sigma\in\mathcal D_1\).
If
\[
(1-t_0)g+t_0\sigma\in\mathcal A
\]
for one \(t_0\in(0,1)\), then
\[
(1-t)g+t\sigma\in\mathcal A
\qquad
\text{for every }0<t<1.
\]
Consequently,
\begin{equation}
\widehat{\mathcal M}_{g,t}
=
\widehat{\mathcal M}_{g},
\qquad
\mathcal M_{g,t}
=
\mathcal M_g
\qquad
(0<t<1).
\label{eq:sm-mixing-rate-equality}
\end{equation}
\end{proposition}

\begin{proof}
{
Let
\[
\rho_{t_0}=(1-t_0)g+t_0\sigma.
\]
Because \(\rho_{t_0}\) is alchemical, state universality provides a
free channel \(C\) such that
\[
C(\rho_{t_0})=g.
\]
Since \(0<t_0<1\),
\[
\operatorname{supp}\rho_{t_0}
=
\operatorname{supp}g+\operatorname{supp}\sigma.
\]
By the observation above, the same channel therefore satisfies
\[
C(g)=C(\sigma)=g.
\]
Hence, for every \(0\le t\le1\),
\[
C((1-t)g+t\sigma)=g.
\]
Conversely, because \(g\) is alchemical, it can be freely converted
into \((1-t)g+t\sigma\).  Thus each mixture is freely
interconvertible with \(g\), and the free-interconversion property
summarized in Sec.~\ref{sec:sm-framework} implies that it is
alchemical.  Equation~\eqref{eq:sm-mixing-rate-equality} follows.
}
\end{proof}

\section{Aquaternionicity and proof of extremality}
\label{sec:sm-aquaternionicity}

We now specialize the general framework to aquaternionicity and prove
the extremality result stated in the main text.

\subsection{The aquaternionicity hierarchy}

{Let \(d=2n\ge4\)} and let
\(\Theta:\mathcal H\to\mathcal H\)
be antiunitary with
\[
\Theta^2=-I .
\]
For \(m\ge1\), define
\[
\Theta_m:=\Theta^{\otimes m},
\qquad
\tau_m(X):=\Theta_mX\Theta_m^{-1}.
\]
Then
\(\Theta_m^2=(-1)^mI\), while
\(\tau_m^2=\operatorname{id}\)
on operators.

The aquaternionicity free states are
\begin{equation}
\mathcal F_m^{\rm Aq}
:=
\left\{
\rho\in\mathcal D_m:
\tau_m(\rho)=\rho
\right\}.
\label{eq:sm-aq-free-states}
\end{equation}

These sets are nonempty, closed, and convex and contain the maximally
mixed state.  Tensor-product and permutation consistency follow
directly from
\(\Theta_{m+n}=\Theta_m\otimes\Theta_n\).
Moreover,
\begin{equation}
\operatorname{Tr}_n\circ\tau_{m+n}
=
\tau_m\circ\operatorname{Tr}_n ,
\label{eq:sm-partial-trace-cov}
\end{equation}
so every marginal of a free state is free.
Finally, \(\mathcal F_1^{\rm Aq}\neq\mathcal D_1\), because no pure
state is invariant under \(\tau_1\), as follows from Kramers
orthogonality below.
Thus the hierarchy satisfies the assumptions of
Sec.~\ref{sec:sm-framework}.
{For \(d=2\), the same fixed-point construction gives only \(\mathcal F_1^{\rm Aq}=\{I/2\}\) by Kramers degeneracy and is therefore excluded by the nontriviality convention.}

For later use, note that every vector obeys
\begin{equation}
\langle\Theta\psi|\psi\rangle=0 .
\label{eq:sm-kramers-orthogonality}
\end{equation}
Indeed, antiunitarity gives
\[
\langle\Theta\psi|\Theta^2\psi\rangle
=
\langle\Theta\psi|\psi\rangle,
\]
whereas \(\Theta^2=-I\) makes the left-hand side equal to
\(-\langle\Theta\psi|\psi\rangle\).

\subsection{Free operations and covariance}

{For aquaternionicity, we define an instrument to be free when each branch is \(\Theta\)-covariant. We now show that this concrete covariance definition coincides with the free-ancilla closure of the CFK class introduced in Sec.~\ref{sec:sm-framework}.}

\begin{lemma}[Completely free branches and covariance]
\label{lem:sm-aq-cfree-cov}
Let
\(\mathcal E:m\to \ell\)
be completely positive.
Then
\[
(\mathcal E\otimes\operatorname{id}_r)
(\mathcal C_{m+r}^{\rm Aq})
\subseteq
\mathcal C_{\ell+r}^{\rm Aq}
\qquad
(\forall r\ge0)
\]
if and only if
\begin{equation}
\mathcal E\circ\tau_m
=
\tau_\ell\circ\mathcal E .
\label{eq:sm-aq-covariance}
\end{equation}
\end{lemma}

\begin{proof}
If Eq.~\eqref{eq:sm-aq-covariance} holds, then
\(\mathcal E\otimes\operatorname{id}_r\)
is covariant with respect to
\(\tau_{m+r}\) and \(\tau_{\ell+r}\), and hence maps every positive
fixed point to a positive fixed point.

Conversely, let
\[
V_m^\pm
:=
\{X=X^\dagger:\tau_m(X)=\pm X\}.
\]
The real span of the free cone is \(V_m^+\): for any
\(X\in V_m^+\), adding a sufficiently large multiple of the identity
makes it positive.
Complete free-state preservation therefore implies
\(\mathcal E(V_m^+)\subseteq V_\ell^+\).

Now take \(X\in V_m^-\), and choose a nonzero
\(Y\in V_1^-\).
Both \(X\) and \(Y\) are traceless.
For sufficiently small real \(t\),
\[
\omega_t
=
\frac{I_m}{d^m}\otimes\frac{I}{d}
+tX\otimes Y
\]
is a positive normalized fixed point of \(\tau_{m+1}\).
Applying
\(\mathcal E\otimes\operatorname{id}_1\)
to \(\omega_t\) and \(\omega_0\) and taking their difference shows that
\(\mathcal E(X)\otimes Y\)
is fixed by \(\tau_{\ell+1}\).
Since \(\tau_1(Y)=-Y\), this implies
\(\tau_\ell(\mathcal E(X))=-\mathcal E(X)\).
Thus \(\mathcal E\) preserves both eigenspaces of \(\tau_m\), which is
equivalent to Eq.~\eqref{eq:sm-aq-covariance}.
\end{proof}

To connect covariance with the CFK definition, define for an operator
\(A:\mathcal H_m\to\mathcal H_\ell\)
\[
A^\#
:=
\Theta_\ell A\Theta_m^{-1}.
\]

{
Applying this operation twice gives
\[
(A^\#)^\#=(-1)^{m+\ell}A.
\]
For a nonzero CFK Kraus operator \(A\), Lemma~\ref{lem:sm-aq-cfree-cov} applied to the single-Kraus map \(\operatorname{Ad}_A\) implies covariance. Hence \(\operatorname{Ad}_{A^\#}=\operatorname{Ad}_A\), so minimality of the single-Kraus representation gives \(A^\#=e^{i\phi}A\) for some phase \(e^{i\phi}\). Since \(\#\) is conjugate linear, this relation implies \((A^\#)^\#=A\). The identity above therefore shows that a nonzero raw CFK Kraus operator can exist only when \(m\equiv\ell\pmod 2\). This arity-parity obstruction is precisely why the free-ancilla closure is needed to represent opposite-parity covariant branches.
}

\begin{lemma}[Covariant maps of equal parity are CFK]
\label{lem:sm-aq-same-parity}
Suppose \(m\equiv\ell\pmod2\).
Every covariant completely positive map
\(\mathcal E:m\to\ell\)
admits a CFK Kraus representation.
\end{lemma}

\begin{proof}
Take any Kraus representation
\(\mathcal E(X)=\sum_\alpha A_\alpha XA_\alpha^\dagger\).
Covariance implies that
\(\{A_\alpha^\#\}_\alpha\)
is another Kraus representation of the same map.
Since \(m\equiv\ell\pmod2\),
\((A_\alpha^\#)^\#=A_\alpha\).
Define
\[
K_{\alpha,0}
=
\frac{A_\alpha+A_\alpha^\#}{2},
\qquad
K_{\alpha,1}
=
\frac{A_\alpha-A_\alpha^\#}{2i}.
\]
Then \(K_{\alpha,b}^\#=K_{\alpha,b}\), and
\[
\sum_{\alpha,b}\operatorname{Ad}_{K_{\alpha,b}}
=
\frac12
\sum_\alpha
\left(
\operatorname{Ad}_{A_\alpha}
+
\operatorname{Ad}_{A_\alpha^\#}
\right)
=
\mathcal E.
\]
The relation
\(K_{\alpha,b}\Theta_m=\Theta_\ell K_{\alpha,b}\)
continues to hold after tensoring with any reference identity, so each
\(K_{\alpha,b}\) satisfies the CFK condition.
\end{proof}

\begin{proposition}[{Covariance equals the general CFK free class}]
\label{prop:sm-aq-free-cov}
{For the hierarchy
\(\{\mathcal F_m^{\rm Aq}\}_m\), an instrument
\(\{\mathcal E_x:m\to\ell\}_x\)
is branchwise \(\Theta\)-covariant if and only if it belongs to the free-ancilla closure of the CFK instruments defined in Sec.~\ref{sec:sm-framework}. Equivalently, each branch satisfies}
\[
\mathcal E_x\circ\tau_m
=
\tau_\ell\circ\mathcal E_x .
\]
\end{proposition}

\begin{proof}
{Every branch in the free-ancilla closure of the CFK class is completely free-state preserving by the closure properties summarized in Sec.~\ref{sec:sm-framework}, and is therefore covariant by}
Lemma~\ref{lem:sm-aq-cfree-cov}.

Conversely, suppose each branch is covariant.
If
\(m\equiv\ell\pmod2\),
Lemma~\ref{lem:sm-aq-same-parity} gives a CFK realization directly.

If
\(m\not\equiv\ell\pmod2\),
append one elementary free ancilla
\(f=I/d\) and define
\[
\widehat{\mathcal E}_x(Y)
:=
\mathcal E_x(\operatorname{Tr}_F Y),
\qquad
\widehat{\mathcal E}_x:m+1\to\ell .
\]
Equation~\eqref{eq:sm-partial-trace-cov} implies that
\(\widehat{\mathcal E}_x\) is covariant.
Its input and output sizes now have the same parity, so
Lemma~\ref{lem:sm-aq-same-parity} gives a CFK realization.
Finally,
\[
\widehat{\mathcal E}_x(X\otimes f)
=
\mathcal E_x(X),
\]
which is exactly the free-ancilla realization of
Eq.~\eqref{eq:sm-free-ancilla}.
\end{proof}

\begin{corollary}[Free discard]
\label{cor:sm-aq-free-discard}
For the aquaternionicity hierarchy, the partial trace over any
collection of elementary systems is a free channel.
\end{corollary}

\begin{proof}
Equation~\eqref{eq:sm-partial-trace-cov} gives
\[
\operatorname{Tr}_R\circ\tau_{m+r}
=
\tau_m\circ\operatorname{Tr}_R .
\]
Thus the partial trace is covariant.
Proposition~\ref{prop:sm-aq-free-cov} therefore implies that it is free.
\end{proof}

{\subsection{Isotropic subspaces and reusable simulation}}

Define
\[
\Omega(x,y)
:=
\langle\Theta x|y\rangle .
\]
Because \(\Theta\) is antiunitary, \(\Omega\) is complex bilinear; by
Eq.~\eqref{eq:sm-kramers-orthogonality} it is alternating and
nondegenerate.
A subspace
\(S\subseteq\mathcal H\)
is \emph{isotropic} when
\[
\Omega(S,S)=0,
\]
equivalently,
\[
S\perp\Theta S.
\]
{
Every isotropic subspace has complex dimension at most
\(n=d/2\).
We call an \(n\)-dimensional isotropic subspace \(L\) a
\emph{maximal isotropic subspace} (equivalently, a Lagrangian
subspace). Any isotropic subspace can be extended to such an \(L\)
by completing Kramers pairs, and then
\[
\mathcal H=L\oplus\Theta L.
\]
For any subspace \(L\subseteq\mathcal H\), we write
\[
\mathcal D(L)
:=
\{\rho\in\mathcal D_1:\operatorname{supp}\rho\subseteq L\}.
\]
}

{
The next proposition is the central operational property of a maximal
isotropic sector.  Once such a subspace \(L\) is fixed, every state
{supported on \(L\) can be used as a reusable resource for simulating an
arbitrary finite-system quantum instrument.  Importantly, the same free
processor works uniformly for all \(\omega\in\mathcal D(L)\): it may
depend on \(L\) and on the target instrument, but not on the particular
resource state, and it returns that resource state unchanged.}
}

\begin{proposition}[Reusable universal simulation on a maximal isotropic sector]
\label{prop:sm-aq-uniform-simulation}
{
Let \(L\subseteq\mathcal H\) be a maximal isotropic subspace, i.e.,
\[
\dim_{\mathbb C}L=n,\qquad L\perp\Theta L.
\]
For every finite \(m\ge1\) and every quantum instrument
\(\{\mathcal N_x:m\to m\}_x\), there exists a free instrument
\(\{\Lambda_{x,L}:m+1\to m+1\}_x\) such that
\begin{equation}
\Lambda_{x,L}(\rho\otimes\omega)
=
\mathcal N_x(\rho)\otimes\omega
\label{eq:sm-aq-reusable-simulation}
\end{equation}
for every
\(\rho\in\mathcal D_m\),
\(\omega\in\mathcal D(L)\),
and \(x\).
}
\end{proposition}

\begin{proof}
Let
\(P=P_L\) and
\(\bar P=P_{\Theta L}\), so that
\(P+\bar P=I\) and \(P\bar P=0\).
Choose a Kraus representation
\[
\mathcal N_x(X)
=
\sum_\mu
A_{x\mu}XA_{x\mu}^\dagger,
\qquad
\sum_{x,\mu}
A_{x\mu}^\dagger A_{x\mu}=I_m ,
\]
and define
\[
A_{x\mu}^\#
:=
\Theta_mA_{x\mu}\Theta_m^{-1},
\qquad
K_{x\mu}
:=
A_{x\mu}\otimes P
+
A_{x\mu}^\#\otimes\bar P .
\]
Orthogonality of the two sectors gives
\[
\sum_{x,\mu}
K_{x\mu}^\dagger K_{x\mu}
=
I_m\otimes(P+\bar P)
=
I_{m+1}.
\]
Moreover,
\((A_{x\mu}^\#)^\#=A_{x\mu}\)
and
\(\tau_1(P)=\bar P\), so
\[
\Theta_{m+1}K_{x\mu}\Theta_{m+1}^{-1}
=
K_{x\mu}.
\]
Thus every \(K_{x\mu}\) is CFK.

Set
\[
\Lambda_{x,L}
=
\sum_\mu
\operatorname{Ad}_{K_{x\mu}} .
\]
For
\(\omega\in\mathcal D(L)\),
we have
\(\omega=P\omega P\) and
\(\bar P\omega=\omega\bar P=0\), from which
Eq.~\eqref{eq:sm-aq-reusable-simulation} follows directly.
\end{proof}

{
\begin{corollary}[States in a maximal isotropic sector are alchemical]
\label{cor:sm-aq-lagrangian-alchemical}
Every state
\(\omega\in\mathcal D(L)\)
supported on a maximal isotropic subspace \(L\) is alchemical.
\end{corollary}
}

\begin{proof}
Proposition~\ref{prop:sm-aq-uniform-simulation} gives, for every
finite-system target instrument, a free reusable processor satisfying
\[
\Lambda_{x,L}(\rho\otimes\omega)
=
\mathcal N_x(\rho)\otimes\omega .
\]
{Composing each branch with the free partial trace over the resource}
register, Corollary~\ref{cor:sm-aq-free-discard}, gives
\[
\widetilde\Lambda_{x,L}(\rho\otimes\omega)
=
\mathcal N_x(\rho),
\]
which is precisely the universal-simulation condition of
{Definition~\ref{def:sm-alchemical}, with one resource copy.}

It remains to verify replication.
Choose as target channel the replacement channel
\[
\mathcal N_\omega(X)=\operatorname{Tr}(X)\omega .
\]
Since \(I/d\in\mathcal F_1^{\rm Aq}\), define
\[
\mathcal R_\omega(X)
:=
\Lambda_{\mathcal N_\omega,L}(I/d\otimes X).
\]
The addition of the free data state and the processor are both free,
so \(\mathcal R_\omega\) is a free channel.
For \(X=\omega\),
\[
\mathcal R_\omega(\omega)
=
\omega\otimes\omega.
\]

Finally, \(\omega\) is nonfree.
If \(\tau_1(\omega)=\omega\), then
\[
\operatorname{supp}\omega
=
\Theta\,\operatorname{supp}\omega.
\]
The left-hand side is contained in \(L\), while the right-hand side is
contained in \(\Theta L\).
Since \(L\perp\Theta L\), this is impossible for a nonzero density
operator.
Hence \(\omega\notin\mathcal F_1^{\rm Aq}\), and \(\omega\) is
alchemical.
\end{proof}

\subsection{Characterization of alchemical states in aquaternionicity}

\begin{proposition}[Alchemical support criterion]
\label{prop:sm-aq-support-criterion}
For a single-system state \(\omega\),
\begin{equation}
\omega\in{\mathcal A^{\rm Aq}}
\quad\Longleftrightarrow\quad
\operatorname{supp}\omega
\perp
\Theta\,\operatorname{supp}\omega .
\label{eq:sm-aq-support-criterion}
\end{equation}
Equivalently,
\[
{\mathcal A^{\rm Aq}}
=
{\bigcup_{L\ {\rm maximal\ isotropic}}}
\mathcal D(L).
\]
In particular, every alchemical state has rank at most
\(d/2\).
\end{proposition}

\begin{proof}
Suppose first that \(\omega\) is alchemical.
{By state universality, choose a free channel}
\(\Lambda\) mapping \(\omega\) to an arbitrary pure state
\(g=\ket{\psi}\!\bra{\psi}\).
By Proposition~\ref{prop:sm-aq-free-cov},
\[
\Lambda(\tau_1(\omega))
=
\tau_1(g).
\]
Kramers orthogonality gives
\(g\perp\tau_1(g)\), and therefore
\[
\frac12
\|g-\tau_1(g)\|_1
=
1.
\]
Contractivity of trace distance under \(\Lambda\) yields
\[
1
\le
\frac12
\|\omega-\tau_1(\omega)\|_1
\le1.
\]
{
Hence equality holds.  Unit trace distance is equivalent to perfect
distinguishability, i.e., orthogonality of the supports, and therefore
\[
\operatorname{supp}\omega
\perp
\operatorname{supp}\tau_1(\omega)
=
\Theta\,\operatorname{supp}\omega.
\]
}

Conversely, suppose
\(S=\operatorname{supp}\omega\)
is isotropic.
{Extend \(S\) to a maximal isotropic subspace \(L\supseteq S\).}
Thus
\(\omega\in\mathcal D(L)\), and
Corollary~\ref{cor:sm-aq-lagrangian-alchemical}
implies
\(\omega\in{\mathcal A^{\rm Aq}}\).

Finally, \(S\perp\Theta S\) and
\(\dim_{\mathbb C}S=\dim_{\mathbb C}\Theta S\) imply
\(2\dim_{\mathbb C}S\le d\), giving the rank bound.
{The union representation follows from the isotropic-extension property established above.}
\end{proof}

{
\begin{corollary}[Every pure state is alchemical]
\label{cor:sm-aq-all-pure-alchemical}
Every pure state in \(\mathcal P_1\) is alchemical in
aquaternionicity.
\end{corollary}

\begin{proof}
For every unit vector \(\psi\), Kramers orthogonality gives
\[
\langle\Theta\psi|\psi\rangle=0.
\]
Hence the one-dimensional support
\(\mathbb C\psi\) is isotropic, and
Proposition~\ref{prop:sm-aq-support-criterion} applies.
\end{proof}
}

\subsection{Saturation of the universal mixing bound}

\begin{proposition}[Uniform saturation]
\label{prop:sm-aq-saturation}
Let
\(g=\ket{\psi}\!\bra{\psi}\)
be any pure state.
Then
\begin{equation}
\mathcal M_g^{\rm Aq}
=
\mathcal P((\Theta\psi)^\perp),
\qquad
\dim\mathcal M_g^{\rm Aq}
=
2d-4.
\label{eq:sm-aq-saturation}
\end{equation}
More generally,
\[
\widehat{\mathcal M}_g^{\rm Aq}
=
{\mathcal A^{\rm Aq}}
\cap
\mathcal D((\Theta\psi)^\perp).
\]
\end{proposition}

\begin{proof}
Let
\(h=\ket{\phi}\!\bra{\phi}\).
For \(0<t<1\),
\[
\operatorname{supp}[(1-t)g+th]
=
\operatorname{span}_{\mathbb C}\{\psi,\phi\}.
\]
By Proposition~\ref{prop:sm-aq-support-criterion}, this mixture is
alchemical if and only if this two-dimensional support is isotropic.

The self-conditions
\[
\langle\Theta\psi|\psi\rangle=0,
\qquad
\langle\Theta\phi|\phi\rangle=0
\]
hold automatically by Kramers orthogonality.
The only remaining condition is therefore
\[
\langle\Theta\psi|\phi\rangle=0.
\]
Hence
\[
\mathcal M_g^{\rm Aq}
=
\mathcal P((\Theta\psi)^\perp).
\]
The subspace
\((\Theta\psi)^\perp\)
has complex dimension \(d-1\), so its pure-state space is
\(\mathbb{CP}^{d-2}\), of real dimension
\(2d-4\).

For a mixed state \(\sigma\), the same support criterion shows that
\((1-t)g+t\sigma\)
is alchemical exactly when
\(\operatorname{supp}\sigma\) is isotropic and contained in
\((\Theta\psi)^\perp\), which gives the final statement.
\end{proof}

\section{Characterization of single-system aquaternionicity}
\label{sec:sm-characterization}

We now prove the converse direction of the main result.
Throughout this section, consider a nontrivial resource theory
satisfying the assumptions of Sec.~\ref{sec:sm-framework} and suppose
that

\begin{enumerate}
\item[(U1)]
every pure state \(g\in\mathcal P_1\) is alchemical;

\item[(U2)]
for every pure state \(g\),
\begin{equation}
\dim\mathcal M_g=2d-4 .
\label{eq:sm-uniform-saturation}
\end{equation}
\end{enumerate}

{By (U1) and state universality, every pair of pure states can be freely converted into one another.}

\subsection{Uniform maximal overlap and the partner map}

{
We now show that the maximal free overlap \(s_g\) introduced in the
proof of Theorem~\ref{thm:sm-universal-mixing-bound} is independent of
the pure state \(g\).
}

\begin{lemma}[Uniform maximal free overlap]
\label{lem:sm-constant-overlap}
There exists a constant \(0<c<1\) such that
\[
s_g=c
\qquad
(\forall g\in\mathcal P_1).
\]
Moreover, every free state \(\sigma\in\mathcal F_1\) satisfies
\[
\sigma\le cI .
\]
\end{lemma}

\begin{proof}
Suppose that a free channel \(\Lambda\) maps a pure state \(g\) to a
pure state \(h\).
Lemma~\ref{lem:sm-effect-support}, applied to
\(\Lambda(g)=h\), gives
\(\Lambda^\dagger(h)\ge g\).
Hence, for every free state \(\sigma\),
\[
s_h
\ge
\operatorname{Tr}[h\Lambda(\sigma)]
=
\operatorname{Tr}[\Lambda^\dagger(h)\sigma]
\ge
\operatorname{Tr}(g\sigma).
\]
Maximizing over \(\sigma\in\mathcal F_1\) gives
\(s_h\ge s_g\).
Since all pure states are mutually freely convertible, the reverse
inequality also holds, so \(s_g=c\) is independent of \(g\).

Nonemptiness of \(\mathcal F_1\) implies \(c>0\).
If \(c=1\), every pure state would be free, and convexity would imply
\(\mathcal F_1=\mathcal D_1\), contradicting nontriviality.
Thus \(0<c<1\).

Finally, for every unit vector \(\psi\) and every
\(\sigma\in\mathcal F_1\),
\(\langle\psi|\sigma|\psi\rangle\le c\).
Therefore \(\sigma\le cI\).
\end{proof}

\begin{proposition}[Unique partner state]
\label{prop:sm-partner}
For every pure state \(g\), there exists a unique pure state
\(\kappa(g)\) orthogonal to \(g\) such that
\begin{equation}
{
\operatorname*{arg\,max}_{\sigma\in\mathcal F_1}
\operatorname{Tr}(g\sigma)
}
=
\{\sigma_g^*\},
\qquad
\sigma_g^*
=
cg+(1-c)\kappa(g).
\label{eq:sm-partner-optimizer}
\end{equation}
Moreover,
\[
c\ge\frac12 .
\]
\end{proposition}

\begin{proof}
Choose any optimizer \(\sigma^*\) for \(g\).
{Using the notation introduced in the proof of Theorem~\ref{thm:sm-universal-mixing-bound},}
\[
C_g
=
\left.
Q_g\sigma^*Q_g
\right|_{\psi^\perp},
\qquad
g=\ket{\psi}\!\bra{\psi}.
\]
By Eq.~\eqref{eq:sm-fixed-t-bound} and the saturation assumption
\eqref{eq:sm-uniform-saturation},
\[
2d-4
=
\dim\mathcal M_g
\le
2(d-\operatorname{rank}C_g)-2 .
\]
Since \(C_g\neq0\), we must have
\(\operatorname{rank}C_g=1\).

Lemma~\ref{lem:sm-constant-overlap} gives
\(cI-\sigma^*\ge0\), while
\(\operatorname{Tr}[g(cI-\sigma^*)]=0\).
Hence
\(\sigma^*\psi=c\psi\), so the off-diagonal block between
\(\mathbb C\psi\) and \(\psi^\perp\) vanishes.
Because the compression to \(\psi^\perp\) has rank one,
\[
\sigma^*
=
cg+(1-c)h
\]
for some pure state \(h\perp g\).

Suppose two distinct pure states \(h_1,h_2\) occurred in this way.
{Applying the support inclusion~\eqref{eq:sm-common-support} separately to the two optimizers gives}
\[
\mathcal M_g
\subseteq
\mathcal P(h_1^\perp\cap h_2^\perp).
\]
For \(d\ge3\), the right-hand side has dimension at most
\(2d-6\), contradicting
Eq.~\eqref{eq:sm-uniform-saturation}.
For \(d=2\), the ray orthogonal to \(g\) is unique.
Thus the optimizer and \(h=\kappa(g)\) are unique.

Since the second eigenvalue of
\(\sigma_g^*\) is \(1-c\), while every free state is bounded above by
\(cI\), we have \(1-c\le c\), hence \(c\ge1/2\).
\end{proof}

\begin{lemma}[Maximal-overlap inequality on the free cone]
\label{lem:sm-free-cone-overlap}
For every pure state \(g\) and every
\(W\in\mathcal C_1\),
\begin{equation}
\operatorname{Tr}(gW)
\le
c\,\operatorname{Tr}W .
\label{eq:sm-free-cone-overlap}
\end{equation}
\end{lemma}

\begin{proof}
If \(W=0\), the statement is immediate.
Otherwise write
\[
W=(\operatorname{Tr}W)\sigma
\]
with \(\sigma\in\mathcal F_1\).
By Lemma~\ref{lem:sm-constant-overlap},
\[
\operatorname{Tr}(g\sigma)\le c.
\]
Multiplying by \(\operatorname{Tr}W\) gives
Eq.~\eqref{eq:sm-free-cone-overlap}.
\end{proof}

At this stage \(\kappa\) is defined only as a map on pure states; no
linearity or continuity is assumed.

\begin{lemma}[Transport of the partner state]
\label{lem:sm-partner-transport}
{
Let \(\Lambda\) be a free channel and let \(g,h\in\mathcal P_1\)
satisfy
\begin{equation}
\Lambda(g)=h.
\label{eq:sm-partner-transport-assumption}
\end{equation}
Then
\[
\Lambda(\kappa(g))
=
\kappa(h).
\]
}
Moreover, if
\(\mathcal R_g\) is any free replicator for \(g\), then
\begin{equation}
\mathcal R_g(\kappa(g))
=
\kappa(g)^{\otimes2}.
\label{eq:sm-partner-replication}
\end{equation}
\end{lemma}

\begin{proof}
{
Since \(\sigma_g^*\) is free and \(\Lambda\) is a free channel,
\[
\Lambda(\sigma_g^*)
=
c\,\Lambda(g)+(1-c)\Lambda(\kappa(g))
=
ch+(1-c)\Lambda(\kappa(g))
\]
is again free, where
Eq.~\eqref{eq:sm-partner-transport-assumption} was used in the second
equality.
Its overlap with \(h\) satisfies
\[
\operatorname{Tr}\!\left[
h\,\Lambda(\sigma_g^*)
\right]
=
c+(1-c)\operatorname{Tr}\!\left[
h\,\Lambda(\kappa(g))
\right]
\ge c.
\]
On the other hand, \(c\) is the maximal overlap of any free state
with the pure state \(h\).  Since
\(\Lambda(\sigma_g^*)\in\mathcal F_1\), we also have
\[
\operatorname{Tr}\!\left[
h\,\Lambda(\sigma_g^*)
\right]
\le c.
\]
Therefore equality holds, so
\(\Lambda(\sigma_g^*)\) is a maximizing free state for \(h\).
By Proposition~\ref{prop:sm-partner}, this maximizer is unique and
equals
\[
\sigma_h^*
=
ch+(1-c)\kappa(h).
\]
Comparing this with
\[
\Lambda(\sigma_g^*)
=
ch+(1-c)\Lambda(\kappa(g))
\]
and using \(c<1\), we obtain
\[
\Lambda(\kappa(g))
=
\kappa(h).
\]
}

For the second statement, write
\(X=\mathcal R_g(\kappa(g))\).
Then
\[
\mathcal R_g(\sigma_g^*)
=
cg^{\otimes2}+(1-c)X
\]
is free.
Taking either marginal and using (F3) gives a free state of the form
\[
cg+(1-c)X_i .
\]
Its overlap with \(g\) is at least \(c\), so uniqueness of the
optimizer implies
\[
cg+(1-c)X_i
=
cg+(1-c)\kappa(g),
\]
hence \(X_i=\kappa(g)\) for both marginals.
A bipartite state with both marginals pure must be their product,
giving Eq.~\eqref{eq:sm-partner-replication}.
\end{proof}

\begin{lemma}[Correlated free states]
\label{lem:sm-correlated-free}
For every \(m\ge1\) and every collection of pure states
\(g_1,\ldots,g_m\),
\begin{equation}
\Omega_{g_1,\ldots,g_m}
:=
c\,g_1\otimes\cdots\otimes g_m
+
(1-c)\,
\kappa(g_1)\otimes\cdots\otimes\kappa(g_m)
\in\mathcal F_m .
\label{eq:sm-correlated-free}
\end{equation}
\end{lemma}

\begin{proof}
Fix a pure state \(q\).
Starting from
\(\sigma_q^*=cq+(1-c)\kappa(q)\)
and repeatedly applying a replicator for \(q\), together with
Lemma~\ref{lem:sm-partner-transport}, produces
\[
cq^{\otimes m}
+
(1-c)\kappa(q)^{\otimes m}
\in\mathcal F_m .
\]
For each \(j\), choose a free channel mapping \(q\) to \(g_j\).
Lemma~\ref{lem:sm-partner-transport} maps
\(\kappa(q)\) simultaneously to \(\kappa(g_j)\).
Applying the tensor product of these channels yields
Eq.~\eqref{eq:sm-correlated-free}.
\end{proof}

\subsection{Measurement duality and linear reconstruction}

\begin{lemma}[Measurement duality]
\label{lem:sm-measurement-duality}
{Fix a pure resource state \(q\).}
For any POVM
\(\{E_x\}_x\), let a free one-copy processor using \(q\) implement an
instrument having this POVM.
{If the same processor with resource state \(\kappa(q)\) induces the POVM}
\(\{E'_x\}_x\), then
\begin{equation}
\operatorname{Tr}[E'_x\kappa(h)]
=
\operatorname{Tr}(E_xh)
\qquad
(\forall h\in\mathcal P_1,\ \forall x).
\label{eq:sm-measurement-duality}
\end{equation}
\end{lemma}

\begin{proof}
Choose, for example, the target instrument
\[
\mathcal N_x(X)
=
\sqrt{E_x}\,X\,\sqrt{E_x}.
\]
{By effective one-copy simulation, let}
\(\{\Psi_x\}_x\) be a free processor implementing this instrument with
{resource state \(q\).}
{With resource state \(\kappa(q)\), define the induced POVM}
\(\{E'_x\}_x\).

Take arbitrary pure states \(h,r\).
By Lemma~\ref{lem:sm-correlated-free},
\[
\Omega_{h,q,r}
=
c\,h\otimes q\otimes r
+
(1-c)\,
\kappa(h)\otimes\kappa(q)\otimes\kappa(r)
\]
is free.
Apply the branch \(\Psi_x\) to the first two systems.
By branchwise complete freeness and marginal consistency, the
unnormalized state on the reference system is a free-cone element
\[
W_x
=
c\,p_xr
+
(1-c)\,p'_x\kappa(r),
\]
where
\[
p_x=\operatorname{Tr}(E_xh),
\qquad
p'_x=\operatorname{Tr}[E'_x\kappa(h)].
\]

{
By Lemma~\ref{lem:sm-free-cone-overlap},
\[
\operatorname{Tr}(rW_x)
\le
c\operatorname{Tr}W_x.
\]
Substituting
\[
W_x=cp_xr+(1-c)p'_x\kappa(r)
\]
and using \(r\perp\kappa(r)\) gives \(p_x\le p'_x\).
Since both POVMs are normalized,
\(\sum_xp_x=\sum_xp'_x=1\), and therefore
\(p_x=p'_x\) for every \(x\).
This proves Eq.~\eqref{eq:sm-measurement-duality}.
}
\end{proof}

{
Choose an informationally complete POVM
\(\{E_x\}_{x=1}^{d^2}\)
whose effects form a real basis of the Hermitian operators.
}

{
\begin{proposition}[Linear extension of \(\kappa\)]
\label{prop:sm-order-automorphism}
The map \(\kappa:\mathcal P_1\to\mathcal P_1\) extends uniquely to an
invertible real-linear map
\[
\kappa:\operatorname{Herm}(\mathcal H)\to
\operatorname{Herm}(\mathcal H).
\]
Both \(\kappa\) and \(\kappa^{-1}\) are positive and trace preserving.
Hence the extension is an affine automorphism of the quantum state space.
\end{proposition}
}

\begin{proof}
{
Let \(\{E'_x\}\) be the POVM obtained from
Lemma~\ref{lem:sm-measurement-duality}.
We first show that the effects \(\{E'_x\}\) are real-linearly
independent.
Suppose
\[
\sum_x a_x E'_x=0
\]
for real coefficients \(a_x\).
Then, for every pure state \(h\),
\[
0
=
\sum_xa_x
\operatorname{Tr}[E'_x\kappa(h)]
=
\sum_xa_x
\operatorname{Tr}(E_xh).
\]
Since pure projectors span the Hermitian space,
\[
\sum_xa_xE_x=0.
\]
The effects \(\{E_x\}\) are real-linearly independent, so
\(a_x=0\) for every \(x\).
Thus \(\{E'_x\}\) is also a real basis of the Hermitian operators.
}

{Define the extension of \(\kappa\) uniquely by}
\begin{equation}
\operatorname{Tr}[E'_x\kappa(X)]
=
\operatorname{Tr}(E_xX)
\qquad
(\forall x).
\label{eq:sm-kappa-definition}
\end{equation}
{Because \(\{E_x\}\) and \(\{E'_x\}\) are both real bases, Eq.~\eqref{eq:sm-kappa-definition} defines a linear bijection; hence \(\kappa\) is invertible.}
{For every pure state \(h\), Eq.~\eqref{eq:sm-measurement-duality}
shows that this extension agrees with the original \(\kappa(h)\).}
Since every positive operator is a nonnegative combination of pure
{projectors, \(\kappa\) is positive and trace preserving.}

{To prove positivity of \(\kappa^{-1}\), apply}
Lemma~\ref{lem:sm-measurement-duality} to an arbitrary binary POVM
\(\{E,I-E\}\).
The corresponding effect \(E'\) satisfies
{\(\kappa^\dagger(E')=E\).}
{Thus \((\kappa^\dagger)^{-1}\) maps every effect, and hence every positive}
operator, to a positive operator.
{Self-duality of the positive cone then implies that \(\kappa^{-1}\) is}
positive.
{Since \(\kappa\) is trace preserving, so is \(\kappa^{-1}\).}
\end{proof}

\subsection{Emergence of the quaternionic antiunitary}

{
Since \(\kappa\) and \(\kappa^{-1}\) are positive and trace preserving by Proposition~\ref{prop:sm-order-automorphism}, the adjoint \(\kappa^\dagger\) is a unital order automorphism of the Hermitian part of \(M_d(\mathbb C)\). Kadison's order-isomorphism theorem~\cite{Kadison1951} therefore implies that \(\kappa^\dagger\) is a Jordan \(*\)-automorphism. For the full matrix algebra, such an automorphism is either a \(*\)-automorphism or a \(*\)-antiautomorphism. Hence there exists a unitary \(U\) such that, on Hermitian operators,
\[
\kappa(X)=UXU^\dagger
\qquad\text{or}\qquad
\kappa(X)=UX^TU^\dagger .
\]
}

\begin{proposition}[Quaternionic antiunitary]
\label{prop:sm-quaternionic-antiunitary}
There exists an antiunitary \(\Theta\) such that
\begin{equation}
{\kappa(X)
=
\Theta X\Theta^{-1},}
\qquad
\Theta^2=-I .
\label{eq:sm-derived-theta}
\end{equation}
In particular, \(d\) is even.
\end{proposition}

\begin{proof}
A unitary implementation is impossible.
Indeed, {if
\(\kappa(X)=UXU^\dagger\),}
an eigenvector \(\psi\) of \(U\) would satisfy
\[
{\kappa(\ket{\psi}\!\bra{\psi})
=
\ket{\psi}\!\bra{\psi},}
\]
contradicting
\(\kappa(g)\perp g\).

{Hence \(\kappa\) is induced by an antiunitary}
\(\Theta=UK\).
Because
{\(\kappa(g)\perp g\) for every pure}
\(g=\ket{\psi}\!\bra{\psi}\),
\[
\langle\psi|\Theta\psi\rangle=0
\qquad
(\forall\psi).
\]
Writing \(z=\overline\psi\), this becomes
\[
z^TU z=0
\qquad
(\forall z\in\mathbb C^d),
\]
which implies
\(U^T=-U\).
Since \(U\) is unitary,
\(U^\dagger=-\overline U\), and therefore
\[
U\overline U=-I.
\]
Thus
\[
\Theta^2
=
UKUK
=
U\overline U
=
-I.
\]
An invertible antisymmetric matrix exists only in even dimension, so
\(d\) is even.
\end{proof}

Set
\[
\tau(X)
:=
\Theta X\Theta^{-1}.
\]
On the real vector space of Hermitian operators,
{\(\tau=\kappa\), and hence}
\(\tau=\tau^{-1}=\tau^\dagger\).
In particular,
\begin{equation}
\operatorname{Tr}[X\tau(Y)]
=
\operatorname{Tr}[\tau(X)Y]
\qquad
(X,Y\ {\rm Hermitian}).
\label{eq:sm-tau-self-adjoint}
\end{equation}

The transformed POVM in
Lemma~\ref{lem:sm-measurement-duality} is therefore fixed uniquely:
\[
E'_x=\tau(E_x).
\]

\subsection{Recovery of the free-state set}

\begin{proposition}[All free states are quaternionic fixed points]
\label{prop:sm-free-invariant}
Every
\(\sigma\in\mathcal F_1\)
satisfies
\[
\tau(\sigma)=\sigma.
\]
\end{proposition}

\begin{proof}
Take arbitrary pure states \(q,r\).
Lemma~\ref{lem:sm-correlated-free}, together with
\(\kappa=\tau\), gives
\[
\Omega_{q,r}
=
cq\otimes r
+
(1-c)\tau(q)\otimes\tau(r)
\in\mathcal F_2 .
\]
Hence
\(\sigma\otimes\Omega_{q,r}\)
is free.

{Use the resource state \(q\) to simulate an arbitrary POVM}
\(\{E_x\}\), and apply the corresponding free processor to
{\(\sigma\) and the resource subsystem.}
Using
\(E'_x=\tau(E_x)\), the resulting unnormalized reference state is
\[
W_x
=
c\,a_xr
+
(1-c)b_x\tau(r)
\in\mathcal C_1,
\]
where
\[
a_x=\operatorname{Tr}(E_x\sigma),
\qquad
b_x=\operatorname{Tr}[\tau(E_x)\sigma].
\]
Lemma~\ref{lem:sm-free-cone-overlap}, applied with \(r\), gives
\(a_x\le b_x\).
Normalization of the two POVMs gives equality for every outcome.

Applying this to arbitrary binary POVMs yields
\[
\operatorname{Tr}(E\sigma)
=
\operatorname{Tr}[\tau(E)\sigma]
=
\operatorname{Tr}[E\tau(\sigma)]
\]
for every effect \(E\), where
Eq.~\eqref{eq:sm-tau-self-adjoint} was used in the last step.
Since effects span the Hermitian space,
\(\sigma=\tau(\sigma)\).
\end{proof}

\begin{theorem}[Single-system characterization]
\label{thm:sm-local-characterization}
Under assumptions (U1) and (U2),
{\(d=2n\ge4\)} and
\[
c=\frac12
\]
and
\begin{equation}
\mathcal F_1
=
\left\{
\sigma\in\mathcal D_1:
\Theta\sigma\Theta^{-1}=\sigma
\right\},
\qquad
\Theta^2=-I .
\label{eq:sm-local-characterization}
\end{equation}
Thus the single-system resource theory is aquaternionicity.
\end{theorem}

\begin{proof}
From
Proposition~\ref{prop:sm-partner} and
\(\kappa(g)=\tau(g)\),
\[
\sigma_g^*
=
cg+(1-c)\tau(g)
\]
is free.
Proposition~\ref{prop:sm-free-invariant} therefore implies
\[
cg+(1-c)\tau(g)
=
c\tau(g)+(1-c)g.
\]
Since \(g\) and \(\tau(g)\) are distinct orthogonal pure states,
\(c=1/2\).

We already know
\[
\mathcal F_1
\subseteq
\{\sigma:\tau(\sigma)=\sigma\}.
\]
For the reverse inclusion, let
\(\tau(\sigma)=\sigma\) and choose any pure-state decomposition
\(\sigma=\sum_ip_ig_i\).
Then
\[
\sigma
=
\frac{\sigma+\tau(\sigma)}2
=
\sum_i
p_i\frac{g_i+\tau(g_i)}2
=
\sum_ip_i\sigma_{g_i}^* .
\]
Each
\(\sigma_{g_i}^*\) is free, so convexity gives
\(\sigma\in\mathcal F_1\).
{Finally, \(d=2\) is impossible under the standing nontriviality assumption: when \(\Theta^2=-I\) on \(\mathbb C^2\), Kramers degeneracy makes every \(\Theta\)-invariant Hermitian operator proportional to \(I\), so Eq.~\eqref{eq:sm-local-characterization} would give \(\mathcal F_1=\{I/2\}\). Hence \(d=2n\ge4\).}
\end{proof}

\begin{proposition}[Uniqueness]
\label{prop:sm-theta-unique}
The map \(\tau\) is uniquely determined by the single-system
free-state set.
Any two antiunitaries inducing the same \(\tau\) differ only by an
overall phase.
\end{proposition}

\begin{proof}
For each pure \(g\), the maximizing state \(\sigma_g^*\) is unique,
and
\[
\tau(g)
=
2\sigma_g^*-g .
\]
Hence \(\tau\) is determined on every pure projector and therefore, by
real linearity, on the entire Hermitian space.

If antiunitaries
\(\Theta_1\) and \(\Theta_2\) induce the same action on all operators,
then
\(U=\Theta_2^{-1}\Theta_1\)
is unitary and leaves every ray invariant.
Applying this to an orthonormal basis and to pairwise sums of basis
vectors shows that all eigenphases coincide.
Thus
\(U=e^{i\alpha}I\), i.e.
\[
\Theta_1=e^{i\alpha}\Theta_2 .
\]
\end{proof}

\section{Composite hierarchy}
\label{sec:sm-composite}

We now complete the characterization by determining the free-state sets
on arbitrary finite composite systems.
Let \(\Theta\) be the antiunitary reconstructed in
Sec.~\ref{sec:sm-characterization}, and define
\[
\Theta_m:=\Theta^{\otimes m},
\qquad
\tau_m(X):=\Theta_mX\Theta_m^{-1}.
\]
Then \(\Theta_m^2=(-1)^mI\), while
\(\tau_m^2=\operatorname{id}\) on operators.
We write
\[
\mathcal F_m^{\rm Aq}
:=
\left\{
\rho\in\mathcal D_m:
\tau_m(\rho)=\rho
\right\}.
\]

\subsection{Upper inclusion}

\begin{proposition}[Composite free states are fixed points]
\label{prop:sm-composite-upper}
For every \(m\ge1\),
\begin{equation}
\mathcal F_m
\subseteq
\mathcal F_m^{\rm Aq}.
\label{eq:sm-composite-upper}
\end{equation}
\end{proposition}

\begin{proof}
Fix pure states \(q,r\).
Since \(c=1/2\) and \(\kappa=\tau_1\),
Lemma~\ref{lem:sm-correlated-free} gives
\[
\Omega_{q,\ldots,q,r}
=
\frac12 q^{\otimes m}\otimes r
+
\frac12 \tau_1(q)^{\otimes m}\otimes\tau_1(r)
\in\mathcal F_{m+1}.
\]

Let \(\sigma\in\mathcal F_m\).
By product closure,
\(\sigma\otimes\Omega_{q,\ldots,q,r}\)
is free.
For each \(j=1,\ldots,m\), choose a POVM
\(\{E^{(j)}_{x_j}\}_{x_j}\)
{and a free processor using \(q\) as a resource to simulate it on the}
\(j\)th data system.
By the measurement-duality construction of
Sec.~\ref{sec:sm-characterization}, {the same processor with resource state}
\(\tau_1(q)\) induces the POVM
\(\{\tau_1(E^{(j)}_{x_j})\}_{x_j}\).

Apply the tensor product of these processors to the \(m\) data systems
{and the corresponding resource systems.}
For an outcome
\(\boldsymbol{x}=(x_1,\ldots,x_m)\), branchwise complete freeness and
marginal consistency give an unnormalized free state on the reference
system of the form
\[
W_{\boldsymbol{x}}
=
\frac12 a_{\boldsymbol{x}}\,r
+
\frac12 b_{\boldsymbol{x}}\,\tau_1(r),
\]
where
\[
a_{\boldsymbol{x}}
=
\operatorname{Tr}
\!\left[
\left(
E^{(1)}_{x_1}\otimes\cdots\otimes E^{(m)}_{x_m}
\right)\sigma
\right],
\]
and
\[
b_{\boldsymbol{x}}
=
\operatorname{Tr}
\!\left[
\left(
\tau_1(E^{(1)}_{x_1})\otimes\cdots\otimes
\tau_1(E^{(m)}_{x_m})
\right)\sigma
\right].
\]

Lemma~\ref{lem:sm-free-cone-overlap}, applied once with \(r\) and once
with \(\tau_1(r)\), gives
\[
a_{\boldsymbol{x}}\le b_{\boldsymbol{x}},
\qquad
b_{\boldsymbol{x}}\le a_{\boldsymbol{x}}.
\]
Hence
\[
a_{\boldsymbol{x}}=b_{\boldsymbol{x}}.
\]

Choose each local POVM to be informationally complete.
Their tensor-product effects span the Hermitian operators on
\(\mathcal H_m\).
Since
\[
\tau_m
\left(
E_1\otimes\cdots\otimes E_m
\right)
=
\tau_1(E_1)\otimes\cdots\otimes\tau_1(E_m),
\]
the equality above implies
\[
\operatorname{Tr}(E\sigma)
=
\operatorname{Tr}[E\tau_m(\sigma)]
\]
for a real basis of Hermitian operators \(E\).
Therefore
\(\sigma=\tau_m(\sigma)\).
\end{proof}

\subsection{Global pure-state reachability}

{
We will use the following immediate consequence of alchemicality.
For any pure alchemical state \(q\), any finite \(m\ge1\), and any
pure state \(h\in\mathcal P_m\), there exists a free channel
\[
\mathcal C_h:1\to m
\]
such that
\[
\mathcal C_h(q)=h.
\]
Indeed, choose the replacement channel
\(\mathcal N_h(X)=\operatorname{Tr}(X)h\) as the target simulation
and generate the required finite number of resource copies by
replication.
}

\begin{lemma}[Partner transport to arbitrary finite systems]
\label{lem:sm-global-partner-transport}
Let \(q\in\mathcal P_1\), and let
\(\mathcal C:1\to m\)
be a free channel such that
\[
\mathcal C(q)=h
\]
is pure.
Then
\begin{equation}
\mathcal C(\tau_1(q))
=
\tau_m(h).
\label{eq:sm-global-partner-transport}
\end{equation}
\end{lemma}

\begin{proof}
Choose any pure state \(r\in\mathcal P_1\).
From Lemma~\ref{lem:sm-correlated-free},
\[
\Omega_{q,r}
=
\frac12 q\otimes r
+
\frac12\tau_1(q)\otimes\tau_1(r)
\in\mathcal F_2.
\]
Applying
\(\mathcal C\otimes\operatorname{id}\)
gives
\[
W
=
\frac12h\otimes r
+
\frac12\eta\otimes\tau_1(r)
\in\mathcal F_{m+1},
\qquad
\eta:=\mathcal C(\tau_1(q)).
\]
By Proposition~\ref{prop:sm-composite-upper},
\(W\) is invariant under \(\tau_{m+1}\).
Hence
\[
h\otimes r
+
\eta\otimes\tau_1(r)
=
\tau_m(h)\otimes\tau_1(r)
+
\tau_m(\eta)\otimes r.
\]
Kramers orthogonality gives
\(r\perp\tau_1(r)\).
Comparing the two orthogonal reference blocks therefore yields
\[
\eta=\tau_m(h),
\]
which is Eq.~\eqref{eq:sm-global-partner-transport}.
\end{proof}

\subsection{Full recovery of the composite hierarchy}

\begin{theorem}[Composite characterization]
\label{thm:sm-composite-characterization}
Under assumptions (U1) and (U2), for every finite \(m\ge1\),
\begin{equation}
\mathcal F_m
=
\mathcal F_m^{\rm Aq}
=
\left\{
\rho\in\mathcal D_m:
\Theta_m\rho\Theta_m^{-1}=\rho
\right\}.
\label{eq:sm-composite-characterization}
\end{equation}
Thus the two extremal requirements determine the full aquaternionicity
hierarchy.
\end{theorem}

\begin{proof}
The inclusion
\(\mathcal F_m\subseteq\mathcal F_m^{\rm Aq}\)
was established in
Proposition~\ref{prop:sm-composite-upper}.

For the reverse inclusion, fix any pure alchemical state
\(q\in\mathcal P_1\).
The single-system characterization,
Theorem~\ref{thm:sm-local-characterization}, gives
\[
f_q
:=
\frac12
\left(
q+\tau_1(q)
\right)
\in\mathcal F_1 .
\]

Let \(h\in\mathcal P_m\) be arbitrary.
{
By the reachability property stated above, choose a free channel
\(\mathcal C_h:1\to m\) such that
\(\mathcal C_h(q)=h\).
}
Lemma~\ref{lem:sm-global-partner-transport} then gives
\[
\mathcal C_h(\tau_1(q))
=
\tau_m(h).
\]
Therefore
\begin{equation}
\frac12
\left(
h+\tau_m(h)
\right)
=
\mathcal C_h(f_q)
\in\mathcal F_m .
\label{eq:sm-symmetrized-pure-free}
\end{equation}

Now take an arbitrary
\(\rho\in\mathcal F_m^{\rm Aq}\).
Choose any pure-state decomposition
\[
\rho=\sum_i p_i h_i.
\]
Since
\(\tau_m(\rho)=\rho\),
\[
\rho
=
\frac12
\left(
\rho+\tau_m(\rho)
\right)
=
\sum_i
p_i
\frac12
\left(
h_i+\tau_m(h_i)
\right).
\]
Every term on the right-hand side is free by
Eq.~\eqref{eq:sm-symmetrized-pure-free}, and convexity therefore
implies
\(\rho\in\mathcal F_m\).
Hence
\(\mathcal F_m^{\rm Aq}\subseteq\mathcal F_m\), completing the proof.
\end{proof}

Combining Theorems~\ref{thm:sm-local-characterization} and
\ref{thm:sm-composite-characterization}, the two operational
conditions
\[
\text{every pure state is alchemical},
\qquad
\dim\mathcal M_g=2d-4
\quad
(\forall g\in\mathcal P_1)
\]
characterize aquaternionicity on every finite composite system, up to
the overall phase of the reconstructed antiunitary and a unitary
change of basis.
{Together with Proposition~\ref{prop:sm-aq-free-cov}, this also fixes the associated free-instrument class to be the \(\Theta\)-covariant one.}

\section{Comparison with imaginarity and parity asymmetry}
\label{sec:sm-comparison}

We compare aquaternionicity with the two previously known alchemical
resource theories, imaginarity and parity asymmetry.
{The qubit classification and the all-or-nothing property of exact
simulation were established in Ref.~\cite{InverseDraft}; here we}
derive the higher-dimensional state-space geometry needed for the
comparisons in the main text.

{Throughout this section let \(d=2n\ge4\).}
For imaginarity, let \(K\) denote complex conjugation in a fixed
basis.
For parity asymmetry, let
\[
Z=I_n\oplus(-I_n).
\]
For aquaternionicity, let \(\Theta\) be the antiunitary satisfying
\(\Theta^2=-I\).
We use the shorthand
\[
W_{\rm Im}=K,\qquad
W_{\rm Par}=Z,\qquad
W_{\rm Aq}=\Theta ,
\]
and on \(m\) systems
\[
W_{\nu,m}=W_\nu^{\otimes m}.
\]
The corresponding conjugations of states are denoted by
\(\vartheta_{\nu,m}\).

The free states in the three standard theories satisfy
\[
\vartheta_{\nu,m}(\rho)=\rho.
\]
Their free instruments coincide with the corresponding covariant
instruments.  For aquaternionicity this was proved in
Proposition~\ref{prop:sm-aq-free-cov}; the imaginarity and parity
cases are the standard covariant formulations used in
Ref.~\cite{InverseDraft}.

\subsection{Maximal alchemical supports}

\begin{lemma}[Extension to maximal alchemical supports]
\label{lem:sm-three-maximal-extension}
Let {\(d=2n\ge4\)}, and let
\(\nu\in\{{\rm Im},{\rm Par},{\rm Aq}\}\).
If a subspace \(S\subseteq\mathcal H\) satisfies
\[
S\perp W_\nu S,
\]
then there exists an \(n\)-dimensional subspace \(L\supseteq S\)
such that
\[
L\perp W_\nu L.
\]
Moreover,
\[
\mathcal H=L\oplus W_\nu L.
\]
\end{lemma}

\begin{proof}
{For aquaternionicity, this extension follows from the Kramers-pair completion described in Sec.~\ref{sec:sm-aquaternionicity}.}

For imaginarity, consider the nondegenerate complex symmetric
bilinear form
\[
B(x,y)=x^Ty.
\]
The condition \(S\perp KS\) is exactly
\(B|_{S\times S}=0\).
Hence
\[
S\subseteq S^{\perp_B}.
\]
If \(r=\dim_{\mathbb C}S<n\), the quotient
\[
S^{\perp_B}/S
\]
has complex dimension \(2(n-r)\), and the bilinear form induced by
\(B\) on this quotient is nondegenerate.
Indeed, the radical of \(B|_{S^{\perp_B}}\) is
\[
(S^{\perp_B})^{\perp_B}=S.
\]
Every nondegenerate symmetric bilinear form over \(\mathbb C\) of
dimension at least two contains a nonzero isotropic vector: after a
linear change of basis the form is represented by the identity, and
\(e_1+ie_2\) is isotropic.
Choose such a vector in the quotient and lift it to
\(v\in S^{\perp_B}\).
Then
\[
B(v,v)=0,\qquad
B(v,S)=0,
\]
so
\(S+\mathbb Cv\)
is again totally isotropic.
Repeating this construction extends \(S\) to an
\(n\)-dimensional totally isotropic subspace \(L\).
Since \(L\perp KL\) and both spaces have dimension \(n\),
\[
\mathcal H=L\oplus KL.
\]

For parity asymmetry, write
\[
\mathcal H=\mathcal H_+\oplus\mathcal H_-,
\qquad
Z=I_{\mathcal H_+}\oplus(-I_{\mathcal H_-}),
\]
with
\(\dim\mathcal H_+=\dim\mathcal H_-=n\).
For \(s=(x,y)\in S\), neutrality gives
\[
\|x\|^2-\|y\|^2=0.
\]
Therefore both coordinate projections
\(P_\pm|_S\) are injective.
Let
\[
M_\pm=P_\pm S.
\]
The condition
\(\langle s|Zs'\rangle=0\)
for all \(s,s'\in S\) implies that
\[
U_0:M_+\to M_-,
\qquad
U_0(P_+s)=P_-s,
\]
is an isometry.
Extend \(U_0\) to a unitary
\[
U:\mathcal H_+\to\mathcal H_-.
\]
Then
\[
L
=
\left\{
\frac{x\oplus Ux}{\sqrt2}:
x\in\mathcal H_+
\right\}
\]
contains \(S\), has complex dimension \(n\), and satisfies
\(L\perp ZL\).
Since \(ZL\) also has dimension \(n\),
\[
\mathcal H=L\oplus ZL.
\]
\end{proof}

\subsection{Alchemical support conditions}

\begin{proposition}[Alchemical support criterion]
\label{prop:sm-three-support}
Let \(\rho\) be a single-system state and
\(S=\operatorname{supp}\rho\).
For
\(\nu\in\{{\rm Im},{\rm Par},{\rm Aq}\}\),
\begin{equation}
\rho\in{\mathcal A^\nu}
\quad\Longleftrightarrow\quad
S\perp W_\nu S .
\label{eq:sm-three-support}
\end{equation}
For a pure state
\(g=\ket{\psi}\!\bra{\psi}\), this becomes
\begin{equation}
\begin{array}{ll}
{\rm Im}:&
\psi^T\psi=0,\\[1mm]
{\rm Par}:&
\langle\psi|Z|\psi\rangle=0,\\[1mm]
{\rm Aq}:&
\langle\Theta\psi|\psi\rangle=0,
\end{array}
\label{eq:sm-three-pure-conditions}
\end{equation}
where the aquaternionicity condition holds for every \(\psi\).
\end{proposition}

\begin{proof}
Necessity is common to all three theories.
Choose a pure alchemical state \(g_0\) satisfying
\[
g_0\perp\vartheta_{\nu,1}(g_0).
\]
Such a state exists in each model: for imaginarity one may take
\((e+if)/\sqrt2\) with real orthonormal \(e,f\); for parity asymmetry
one may take \((e+f)/\sqrt2\) with \(e\) and \(f\) in the \(+1\) and
\(-1\) sectors of \(Z\); for aquaternionicity every pure state has
this property.

If \(\rho\) is alchemical, state universality gives a free channel
\(\Lambda\) with
\(\Lambda(\rho)=g_0\).
Covariance implies
\[
\Lambda(\vartheta_{\nu,1}(\rho))
=
\vartheta_{\nu,1}(g_0).
\]
Since the two output states have orthogonal supports, their trace
distance is one.
Contractivity therefore gives
\[
D(\rho,\vartheta_{\nu,1}(\rho))=1,
\]
{and unit trace distance is equivalent to orthogonality of the
supports, i.e.,
\[
S\perp W_\nu S.
\]}

Conversely, suppose \(S\perp W_\nu S\).
By Lemma~\ref{lem:sm-three-maximal-extension}, choose an
\(n\)-dimensional \(L\supseteq S\) with
\[
\mathcal H=L\oplus W_\nu L.
\]
Let
\(P=P_L\) and \(Q=P_{W_\nu L}\).
For a target instrument
\[
\mathcal N_x(X)
=
\sum_\mu A_{x\mu}XA_{x\mu}^\dagger,
\]
define
\[
A_{x\mu}^{\#}
=
W_{\nu,m}A_{x\mu}W_{\nu,m}^{-1},
\]
with the natural conjugate-linear interpretation for the
antiunitary cases, and
\[
K_{x\mu}
=
A_{x\mu}\otimes P
+
A_{x\mu}^{\#}\otimes Q.
\]
The two blocks are orthogonal and exchanged by the relevant symmetry,
so
\[
\sum_{x,\mu}K_{x\mu}^\dagger K_{x\mu}=I
\]
and each branch is covariant.
Hence the resulting instrument is free.
For every
\(\omega\in\mathcal D(L)\),
\[
\Lambda_{x,L}(\rho\otimes\omega)
=
\mathcal N_x(\rho)\otimes\omega.
\]

{The processor returns the resource state unchanged.}
Partial trace is covariant in all three standard theories and hence
{free, so discarding the returned resource state gives the simulator required}
by Definition~\ref{def:sm-alchemical}.
Choosing a replacement channel as the target gives exact replication.
Thus every state supported in \(L\), and in particular \(\rho\), is
alchemical.

The pure-state conditions in
Eq.~\eqref{eq:sm-three-pure-conditions}
are the one-dimensional instances of
Eq.~\eqref{eq:sm-three-support}.
\end{proof}

Thus all pure states are alchemical only in aquaternionicity.
Imaginarity and parity asymmetry select proper subsets of the
pure-state space.

\subsection{Dimensions of the admissible mixing sets}

\begin{proposition}[Admissible mixing dimensions]
\label{prop:sm-three-mixing}
Let \(d=2n\ge4\), and fix a pure alchemical state
\(g=\ket{\psi}\!\bra{\psi}\) in each theory.
Then
\begin{equation}
\dim\mathcal M_g^{\rm Im}=2d-6,\qquad
\dim\mathcal M_g^{\rm Par}=2d-5,\qquad
\dim\mathcal M_g^{\rm Aq}=2d-4.
\label{eq:sm-three-mixing}
\end{equation}
\end{proposition}

\begin{proof}
For
\(h=\ket{\phi}\!\bra{\phi}\)
and \(0<t<1\), the support of
\((1-t)g+th\)
is
\(\operatorname{span}_{\mathbb C}\{\psi,\phi\}\).
Proposition~\ref{prop:sm-three-support} therefore reduces the problem
to the isotropy of this two-dimensional support.

For imaginarity the conditions are
\[
\psi^T\psi=0,\qquad
\psi^T\phi=0,\qquad
\phi^T\phi=0.
\]
Choose a real orthonormal basis such that
\[
\psi=\frac{e_1+ie_2}{\sqrt2}.
\]
The mutual condition
\(\psi^T\phi=0\)
allows us to write
\[
\phi
=
a\psi+\sum_{j=3}^d z_je_j .
\]
The remaining self-condition becomes
\[
\sum_{j=3}^d z_j^2=0.
\]
Inside the projective hyperplane defined by the mutual complex-linear
condition, this is one nontrivial complex equation.
Hence
\[
\dim\mathcal M_g^{\rm Im}=2d-6.
\]

For parity asymmetry, write
\[
\psi=\frac{u+v}{\sqrt2},
\qquad
u\in\mathcal H_+,\quad
v\in\mathcal H_-,
\qquad
\|u\|=\|v\|=1,
\]
where
\(\mathcal H_\pm\)
are the two eigenspaces of \(Z\).
The mutual condition
\(\langle\psi|Z|\phi\rangle=0\)
removes one complex degree of freedom.
After resolving this condition,
\[
\phi=a\psi+x+y,
\]
with
\(x\in\mathcal H_+\cap u^\perp\) and
\(y\in\mathcal H_-\cap v^\perp\).
The self-condition
\(\langle\phi|Z|\phi\rangle=0\)
then imposes one additional real constraint,
\[
\|x\|^2=\|y\|^2.
\]
The regular locus is nonempty and has real dimension \(2d-5\);
the locus where this real constraint is singular is lower-dimensional.
Therefore
\[
\dim\mathcal M_g^{\rm Par}=2d-5.
\]

For aquaternionicity, the self-condition
\(\langle\Theta\phi|\phi\rangle=0\)
is automatic.
Only the complex-linear condition
\[
\langle\Theta\psi|\phi\rangle=0
\]
remains, and hence
\[
\mathcal M_g^{\rm Aq}
=
\mathcal P((\Theta\psi)^\perp)
\simeq
\mathbb{CP}^{d-2},
\]
which has dimension \(2d-4\).
\end{proof}

The separation reflects the number of independent self-conditions
that survive after imposing compatibility with the reference resource:
one additional complex condition for imaginarity, one additional real
condition for parity asymmetry, and none for aquaternionicity.

\subsection{Operational benchmarks inherited from imaginarity}

The enlarged state-space abundance of aquaternionicity does not come
from weakening the operational tasks studied previously for
imaginarity \cite{Wu2021,Wu2024}.
We record three representative counterparts.

\paragraph{Resource distillation.---}
Fix any pure aquaternionicity resource \(g_0\), and for \(N\ge1\)
define
\[
F_N^{\rm Aq}(\rho)
:=
\sup_{\Lambda\in{\mathcal O^{\rm Aq}}(N\to1)}
\operatorname{Tr}
\!\left[
g_0\Lambda(\rho^{\otimes N})
\right].
\]

\begin{proposition}[Optimal multi-copy distillation]
\label{prop:sm-aq-distillation}
For every state \(\rho\) and \(N\ge1\),
\begin{equation}
F_N^{\rm Aq}(\rho)
=
\frac12
+
\frac14
\left\|
\rho^{\otimes N}
-
\tau_1(\rho)^{\otimes N}
\right\|_1 .
\label{eq:sm-aq-distillation}
\end{equation}
In particular, every pure state satisfies
\(F_1^{\rm Aq}=1\).
\end{proposition}

\begin{proof}
Let
\[
\Delta_N
=
\rho^{\otimes N}
-
\tau_1(\rho)^{\otimes N},
\qquad
M_N
=
\frac{I+\operatorname{sgn}\Delta_N}{2}.
\]
Since
\(\tau_N(\Delta_N)=-\Delta_N\),
we have
\(\tau_N(M_N)=I-M_N\).
The measure-and-prepare channel
\[
\Lambda_N(X)
=
\operatorname{Tr}(M_NX)\,g_0
+
\operatorname{Tr}[(I-M_N)X] \,\tau_1(g_0)
\]
is therefore covariant and hence free.
{
Since \(\tau_N(M_N)=I-M_N\), self-adjointness of \(\tau_N\) gives
\[
\operatorname{Tr}\!\left[
(I-M_N)\tau_1(\rho)^{\otimes N}
\right]
=
\operatorname{Tr}\!\left[
M_N\rho^{\otimes N}
\right].
\]
Hence the equal-prior Helstrom success probability for discriminating
\(\rho^{\otimes N}\) from \(\tau_1(\rho)^{\otimes N}\) is precisely
\(\operatorname{Tr}[M_N\rho^{\otimes N}]\), which is also the overlap
of \(\Lambda_N(\rho^{\otimes N})\) with \(g_0\).
The Helstrom formula therefore gives Eq.~\eqref{eq:sm-aq-distillation}.
}

To prove optimality, let
\[
\sigma=\Lambda(\rho^{\otimes N})
\]
for an arbitrary free channel \(\Lambda\), and set
\[
p=\operatorname{Tr}(g_0\sigma),
\qquad
q=\operatorname{Tr}[\tau_1(g_0)\sigma].
\]
Since
\(g_0\perp\tau_1(g_0)\),
\[
p+q\le1.
\]
Moreover, covariance of \(\Lambda\) gives
\[
\tau_1(\sigma)
=
\Lambda\!\left(\tau_1(\rho)^{\otimes N}\right).
\]
Using the self-adjointness of \(\tau_1\) with respect to the trace
pairing,
\[
p-q
=
\operatorname{Tr}
\!\left[
g_0\bigl(\sigma-\tau_1(\sigma)\bigr)
\right].
\]
Since \(g_0\) is an effect and
\(\operatorname{Tr}[\sigma-\tau_1(\sigma)]=0\),
\[
p-q
\le
\frac12
\left\|
\sigma-\tau_1(\sigma)
\right\|_1.
\]
Therefore
\[
2p
=
(p+q)+(p-q)
\le
1+
\frac12
\left\|
\sigma-\tau_1(\sigma)
\right\|_1.
\]
Trace-distance contractivity then gives
\[
p
\le
\frac12
+
\frac14
\left\|
\rho^{\otimes N}
-
\tau_1(\rho)^{\otimes N}
\right\|_1.
\]
This coincides with the value achieved above and proves
Eq.~\eqref{eq:sm-aq-distillation}.

For a pure input \(g\), Kramers orthogonality gives
\(g\perp\tau_1(g)\), so the right-hand side equals one already for
\(N=1\).
\end{proof}

{
\begin{corollary}[Asymptotic distillation]
\label{cor:sm-aq-asymptotic-distillation}
Let \(\rho\) be nonfree, i.e.,
\(\rho\neq\tau_1(\rho)\).
Then the optimal one-copy distillation error satisfies
\begin{equation}
\lim_{N\to\infty}
-\frac1N
\log\!\left(1-F_N^{\rm Aq}(\rho)\right)
=
\xi_{\rm QCB}\!\left(\rho,\tau_1(\rho)\right),
\label{eq:sm-aq-distillation-exponent}
\end{equation}
where
\[
\xi_{\rm QCB}(\rho,\sigma)
:=
-\log
\min_{0\le s\le1}
\operatorname{Tr}\!\left(\rho^s\sigma^{1-s}\right)
\]
is the quantum Chernoff exponent~\cite{Audenaert2007}.

Moreover, in the usual asymptotic resource-distillation sense, the
distillation yield of \(g_0\) is unbounded.  More precisely, for
\emph{any} sequence of positive integers \(m_N\), there exist free
channels
\[
\Gamma_N:N\to m_N
\]
such that
\begin{equation}
D\!\left(
\Gamma_N(\rho^{\otimes N}),
g_0^{\otimes m_N}
\right)
\longrightarrow0.
\label{eq:sm-aq-unbounded-yield}
\end{equation}
Consequently, if the conventional distillation rate is defined as the
supremum of achievable \(m_N/N\), then
\[
R_{\rm dist}(\rho\to g_0)=+\infty.
\]
\end{corollary}

\begin{proof}
By Eq.~\eqref{eq:sm-aq-distillation},
\[
1-F_N^{\rm Aq}(\rho)
=
\frac12
-\frac14
\left\|
\rho^{\otimes N}
-
\tau_1(\rho)^{\otimes N}
\right\|_1,
\]
which is exactly the minimum equal-prior error probability for
discriminating
\(\rho^{\otimes N}\) from
\(\tau_1(\rho)^{\otimes N}\).
Equation~\eqref{eq:sm-aq-distillation-exponent} therefore follows
from the quantum Chernoff theorem~\cite{Audenaert2007}.

For the second statement, let
\(\Lambda_N:N\to1\) be the free channel attaining
Eq.~\eqref{eq:sm-aq-distillation}, and write
\[
\sigma_N:=\Lambda_N(\rho^{\otimes N}).
\]
Since \(F_N^{\rm Aq}(\rho)\to1\), the
Fuchs--van de Graaf inequality gives
\[
D(\sigma_N,g_0)
\le
\sqrt{1-F_N^{\rm Aq}(\rho)}
\longrightarrow0.
\]
By repeated exact replication, for every \(m_N\) there exists a free
channel
\[
\mathcal R_{g_0,m_N}:1\to m_N
\]
satisfying
\[
\mathcal R_{g_0,m_N}(g_0)
=
g_0^{\otimes m_N}.
\]
Thus
\[
\Gamma_N
:=
\mathcal R_{g_0,m_N}\circ\Lambda_N
\]
is free, and contractivity of trace distance yields
\[
D\!\left(
\Gamma_N(\rho^{\otimes N}),
g_0^{\otimes m_N}
\right)
\le
D(\sigma_N,g_0)
\longrightarrow0.
\]
Since \(m_N\) is arbitrary, the asymptotic yield is unbounded.
\end{proof}
}

\paragraph{Assisted state conversion.---}

\begin{proposition}[Deterministic pure-source assisted conversion]
\label{prop:sm-aq-assisted-conversion}
Let \(\Psi^{AB}\) be any finite-dimensional pure bipartite state, and
let Bob's system carry an aquaternionicity structure
\(\Theta_B^2=-I\).
For every target state
\(\sigma^B\), Alice can perform a local measurement and communicate
its outcome to Bob such that Bob, using aquaternionicity-free
operations, prepares
\(\sigma^B\) with unit probability.
\end{proposition}

\begin{proof}
Alice performs a rank-one POVM.
For every nonzero outcome, Bob's conditional state is pure.
Every pure state of Bob's system is alchemical, and hence state
universal.
Bob therefore applies, conditioned on Alice's outcome, a free channel
mapping the conditional pure state to
\(\sigma^B\).
The protocol succeeds for every outcome.
\end{proof}

Thus the probability--fidelity tradeoff appearing in the corresponding
imaginarity task collapses to the deterministic endpoint for pure
sources.

\paragraph{Resource-assisted channel discrimination.---}
Consider the minimal Kramers system with
\(\Theta=i\sigma_yK\), and the two free channels
\begin{equation}
\mathcal N(\rho)
=
\frac12
(\rho+\sigma_y\rho\sigma_y),
\qquad
\mathcal M(\rho)
=
\frac12
(\sigma_x\rho\sigma_x+\sigma_z\rho\sigma_z).
\label{eq:sm-aq-discrimination-channels}
\end{equation}
The only free state of this system is \(I/2\), and
\[
\mathcal N(I/2)=\mathcal M(I/2)=I/2.
\]
Hence no measurement can distinguish the two channels when only a
free probe is available.

By contrast, for
\(g_\pm=\ket{\pm y}\!\bra{\pm y}\),
\[
\mathcal N(g_+)=g_+,
\qquad
\mathcal M(g_+)=g_-,
\qquad
g_+\perp g_-.
\]
An alchemical copy of \(g_+\) can be replicated so that one copy is
{used as the probe and further copies are used to simulate the required}
measurement.
The two channels are therefore perfectly distinguishable with the
resource.
This reproduces the qualitative resource-assisted discrimination
advantage known from imaginarity.

\section{Recalibration-free exact simulation under noise}
\label{sec:sm-noise}

{
We now formalize the recalibration-free simulation mechanism discussed
in the main text.  The setting is the following.  A pure resource
state is prepared and is then affected, before use, by an unknown
member \(\Gamma_\alpha\) of a known noise family
\(\{\Gamma_\alpha\}_{\alpha\in\Xi}\).  The realized value of
\(\alpha\) is not assumed to be known at the time of operation, and
the noise channel itself need not be free.  We ask when one can
compile a single free processor in advance that exactly simulates the
desired target instrument for every resulting noisy resource state,
without estimating \(\alpha\), performing tomography, or
recalibrating the processor.

The section has three steps.  First, using the uniform simulation theorem
proved above, we derive a common-support corollary: a fixed processor
exists whenever all possible noisy resource states are supported on one
common alchemical subspace \(S\).  For aquaternionicity this means
\(S\perp\Theta S\).  Second, to compare
different alchemical theories under the same noise family, we
introduce \(r_\nu(\Gamma)\), the largest dimension of an alchemical
subspace that is invariant under every channel in the family.
This quantity measures the size of a resource sector that survives
the entire noise family while remaining usable by one uniformly
compiled processor.  Third, we evaluate this quantity for two
standard structured noise models---basis dephasing and pairwise
amplitude damping---and compare imaginarity, parity asymmetry, and
aquaternionicity.  We finally show that the resulting advantage is
specific to these structured noise classes and does not define a
universal ordering under arbitrary CPTP noise.
}

\subsection{Uniform simulation on a common support}

For a subspace \(S\subseteq\mathcal H\), write
\[
\mathcal D(S)
:=
\left\{
\omega\in\mathcal D_1:
\operatorname{supp}\omega\subseteq S
\right\}.
\]

{
\begin{corollary}[Uniform simulation on an isotropic support]
\label{cor:sm-uniform-support}
Let \(S\subseteq\mathcal H\) satisfy
\[
S\perp\Theta S.
\]
For every finite \(m\ge1\) and every quantum instrument
\(\{\mathcal N_x:m\to m\}_x\), there exists a single free instrument
\(\{\Lambda_{x,S}:m+1\to m+1\}_x\) such that
\begin{equation}
\Lambda_{x,S}(\rho\otimes\omega)
=
\mathcal N_x(\rho)\otimes\omega
\label{eq:sm-uniform-support}
\end{equation}
for every
\(\rho\in\mathcal D_m\),
\(\omega\in\mathcal D(S)\), and \(x\).
The processor is independent of the particular state
\(\omega\in\mathcal D(S)\).
\end{corollary}
}

\begin{proof}
{Extend \(S\) to a maximal isotropic subspace \(L\supseteq S\).}
Then
\(\mathcal D(S)\subseteq\mathcal D(L)\), and
Proposition~\ref{prop:sm-aq-uniform-simulation} gives a free
instrument
\(\{\Lambda_{x,L}\}_x\) satisfying
Eq.~\eqref{eq:sm-uniform-support} for every
\(\omega\in\mathcal D(L)\).
{Restricting the allowed resource states to}
\(\mathcal D(S)\) proves the claim.
\end{proof}

Consider now a family of noise channels
\(\{\Gamma_\alpha\}_{\alpha\in\Xi}\).

\begin{corollary}[Recalibration-free simulation]
\label{cor:sm-recalibration-free}
Suppose there exist a pure state \(g\) and an isotropic subspace \(S\)
such that
\begin{equation}
\operatorname{supp}\Gamma_\alpha(g)
\subseteq S
\qquad
(\forall\alpha\in\Xi).
\label{eq:sm-noise-common-support}
\end{equation}
Then, for every target instrument, one free processor can be chosen in
advance and used exactly for every noisy resource state
\(\Gamma_\alpha(g)\), independently of \(\alpha\).
In particular, neither estimation of the realized noise parameter nor
tomography of the resulting resource state is required.
\end{corollary}

\begin{proof}
Equation~\eqref{eq:sm-noise-common-support} gives
\(\Gamma_\alpha(g)\in\mathcal D(S)\)
for every \(\alpha\).
{Apply Corollary~\ref{cor:sm-uniform-support}.}
\end{proof}

Because every pure state is alchemical in aquaternionicity, the choice
of the initial pure resource \(g\) is not constrained by
resourcefulness itself.
It may therefore be chosen solely to satisfy the noise-compatibility
condition Eq.~\eqref{eq:sm-noise-common-support}.
This does not assert that such a pair \((g,S)\) exists for every noise
family.

The admissible mixing geometry is the rank-two instance of the same
mechanism.
If
\(g=\ket{\psi}\!\bra{\psi}\)
and
\(h=\ket{\phi}\!\bra{\phi}\in\mathcal M_g^{\rm Aq}\), then
\[
S=\operatorname{span}_{\mathbb C}\{\psi,\phi\}
\]
is isotropic.
Hence a single processor works for the full family
\((1-t)g+th\), \(0\le t\le1\).

\begin{definition}[Noise-invariant alchemical support]
\label{def:sm-noise-invariant-support}
For a noise family
\(\Gamma=\{\Gamma_\alpha\}_{\alpha\in\Xi}\)
and a resource theory \(\nu\), define
\[
r_\nu(\Gamma)
:=
\max_S
\left\{
\dim_{\mathbb C}S:
\Gamma_\alpha(\mathcal D(S))
\subseteq
\mathcal D(S)\ \forall\alpha,\;
\mathcal D(S)\subseteq{\mathcal A^\nu}
\right\}.
\]
If no nonzero such subspace exists, set
\(r_\nu(\Gamma)=0\).
\end{definition}

{
For the two benchmark noise models analyzed below, the resulting
values are summarized in Table~\ref{tab:sm-noise-r-summary}.  Here
{\(d=2n\ge4\)}.
\begin{table}[h]
\centering
\begin{tabular}{c|ccc}
\hline
Noise family
& \(r_{\rm Im}\)
& \(r_{\rm Par}\)
& \(r_{\rm Aq}\)
\\
\hline
Rank-one basis dephasing
& \(0\)
& \(0\)
& \(n\)
\\
Pairwise amplitude damping
& \(\lfloor n/2\rfloor\)
& \(0\)
& \(n\)
\\
\hline
\end{tabular}
\caption{Maximal dimensions of noise-invariant alchemical supports
for the two structured noise models considered in this section.}
\label{tab:sm-noise-r-summary}
\end{table}
}

For aquaternionicity, the condition
\(\mathcal D(S)\subseteq{\mathcal A^{\rm Aq}}\)
is equivalent to
\(S\perp\Theta S\).
For imaginarity and parity asymmetry the corresponding conditions are
\(S\perp KS\) and \(S\perp ZS\), respectively.

\begin{lemma}[Invariant supports and Kraus operators]
\label{lem:sm-kraus-invariant-support}
Let
\[
{\Gamma(X)=\sum_\mu A_\mu XA_\mu^\dagger}
\]
be a quantum channel and \(S\subseteq\mathcal H\) a subspace.
Then
\begin{equation}
\Gamma(\mathcal D(S))
\subseteq
\mathcal D(S)
\quad\Longleftrightarrow\quad
{
A_\mu S\subseteq S
\quad
\text{for every }\mu .
}
\label{eq:sm-kraus-invariant-support}
\end{equation}
The equivalence holds for every Kraus representation of \(\Gamma\).
\end{lemma}

\begin{proof}
{If every \(A_\mu\) preserves \(S\),} the support inclusion is
immediate.

Conversely, assume
\[
\Gamma(\mathcal D(S))
\subseteq\mathcal D(S),
\]
and let \(Q=I-P_S\).
For any \(\psi\in S\),
\[
0
=
Q\Gamma(\ket{\psi}\!\bra{\psi})Q
=
{
\sum_\mu
\ket{QA_\mu\psi}\!\bra{QA_\mu\psi}.
}
\]
Every term on the right-hand side is positive semidefinite.
Their sum can vanish only if
\[
{
QA_\mu\psi=0
\qquad
(\forall\mu).
}
\]
{
Thus
\(A_\mu\psi\in S\)
for every \(\psi\in S\), proving
\(A_\mu S\subseteq S\).
}
\end{proof}

\subsection{A broader dephasing family}

Let
\(\{\Pi_j\}_{j=1}^r\)
be a PVM satisfying
\[
\Theta\Pi_j\Theta^{-1}
=
\Pi_j
\qquad
(\forall j),
\]
and define
\[
\mathsf D_{\Pi,\lambda}(X)
=
\lambda X
+
(1-\lambda)
\sum_j\Pi_jX\Pi_j,
\qquad
0\le\lambda\le1.
\]

\begin{proposition}[Kramers-compatible dephasing]
\label{prop:sm-kramers-dephasing}
For every pure input
\(g=\ket{\psi}\!\bra{\psi}\), define
\[
S_{\psi,\Pi}
:=
\operatorname{span}_{\mathbb C}
\{\Pi_j\psi:\Pi_j\psi\neq0\}.
\]
Then \(S_{\psi,\Pi}\) is isotropic, and
\[
\operatorname{supp}
\mathsf D_{\Pi,\lambda}(g)
\subseteq
S_{\psi,\Pi}
\qquad
(\forall\lambda\in[0,1]).
\]
Consequently, for every target instrument there exists a single free
processor, independent of \(\lambda\), that uses
\(\mathsf D_{\Pi,\lambda}(g)\)
exactly for all
\(\lambda\in[0,1]\).
\end{proposition}

\begin{proof}
Write
\(\psi_j=\Pi_j\psi\).
Because each \(\Pi_j\) is \(\Theta\)-invariant,
\(\Theta\psi_j\in\operatorname{ran}\Pi_j\).
For \(j\neq k\), orthogonality of the PVM gives
\[
\langle\Theta\psi_j|\psi_k\rangle=0,
\]
while for \(j=k\), Kramers orthogonality gives
\[
\langle\Theta\psi_j|\psi_j\rangle=0.
\]
Thus
\(S_{\psi,\Pi}\perp\Theta S_{\psi,\Pi}\).

The output is
\[
\lambda\ket{\psi}\!\bra{\psi}
+
(1-\lambda)
\sum_j
\ket{\psi_j}\!\bra{\psi_j},
\]
whose support lies in
\(S_{\psi,\Pi}\).
Corollary~\ref{cor:sm-recalibration-free} then gives the fixed
processor.
\end{proof}

\subsection{Rank-one basis dephasing}

Let
\(\{e_j,f_j\}_{j=1}^n\)
be orthonormal and satisfy
\[
Ke_j=e_j,\qquad Kf_j=f_j,
\]
\[
Ze_j=e_j,\qquad Zf_j=-f_j,
\]
and
\[
\Theta e_j=f_j,\qquad
\Theta f_j=-e_j.
\]

Consider
\begin{equation}
\Gamma_\lambda^{\rm deph}(X)
=
\lambda X
+
(1-\lambda)
\sum_{j=1}^n
\left(
P_{e_j}XP_{e_j}
+
P_{f_j}XP_{f_j}
\right),
\qquad
0\le\lambda<1.
\label{eq:sm-basis-dephasing}
\end{equation}

\begin{lemma}[Invariant supports for basis dephasing]
\label{lem:sm-dephasing-invariant-support}
For \(0\le\lambda<1\), a subspace \(S\) satisfies
\[
\Gamma_\lambda^{\rm deph}(\mathcal D(S))
\subseteq
\mathcal D(S)
\]
if and only if \(S\) is a coordinate subspace of the basis
\(\{e_j,f_j\}\).
\end{lemma}

\begin{proof}
A Kraus representation is
\[
\sqrt{\lambda}\,I,\qquad
\sqrt{1-\lambda}\,P_{e_j},\qquad
\sqrt{1-\lambda}\,P_{f_j}.
\]
By Lemma~\ref{lem:sm-kraus-invariant-support}, invariance of \(S\)
is therefore equivalent to invariance under every rank-one coordinate
projector.
This holds exactly for subspaces spanned by subsets of the coordinate
basis.
\end{proof}

\begin{proposition}[Basis-dephasing comparison]
\label{prop:sm-basis-dephasing-comparison}
For the family in
Eq.~\eqref{eq:sm-basis-dephasing},
\begin{equation}
r_{\rm Im}=0,\qquad
r_{\rm Par}=0,\qquad
r_{\rm Aq}=n.
\label{eq:sm-basis-dephasing-comparison}
\end{equation}
\end{proposition}

\begin{proof}
Every nonzero coordinate subspace contains a basis vector fixed by
\(K\), and hence cannot satisfy
\(S\perp KS\).
Thus \(r_{\rm Im}=0\).

Likewise, every coordinate basis vector is an eigenvector of \(Z\),
so no nonzero coordinate subspace can satisfy
\(S\perp ZS\).
Hence \(r_{\rm Par}=0\).

For aquaternionicity,
\[
E
:=
\operatorname{span}_{\mathbb C}
\{e_1,\ldots,e_n\}
\]
is a coordinate subspace and satisfies
\[
E\perp\Theta E.
\]
Thus \(r_{\rm Aq}\ge n\).
No alchemical support can exceed dimension \(n=d/2\), so equality
holds.
\end{proof}

For \(d=4\) (\(n=2\)), this yields
\[
(r_{\rm Im},r_{\rm Par},r_{\rm Aq})
=
(0,0,2),
\]
as reported in the main text.

\subsection{Pairwise amplitude damping}

Let
\[
E
=
\operatorname{span}_{\mathbb C}
\{e_1,\ldots,e_n\},
\qquad
F
=
\operatorname{span}_{\mathbb C}
\{f_1,\ldots,f_n\},
\]
and denote their projectors by
\(P_E,P_F\).
For \(0<\gamma\le1\), consider the channel with Kraus operators
\begin{equation}
A_0
=
P_E+\sqrt{1-\gamma}\,P_F,
\qquad
A_j
=
\sqrt{\gamma}\,
\ket{e_j}\!\bra{f_j},
\quad
j=1,\ldots,n.
\label{eq:sm-pair-damping}
\end{equation}

\begin{proposition}[Amplitude-damping comparison]
\label{prop:sm-damping-comparison}
For the channel in
Eq.~\eqref{eq:sm-pair-damping},
\begin{equation}
r_{\rm Im}
=
\left\lfloor\frac n2\right\rfloor,
\qquad
r_{\rm Par}=0,
\qquad
r_{\rm Aq}=n.
\label{eq:sm-damping-comparison}
\end{equation}
\end{proposition}

\begin{proof}
By Lemma~\ref{lem:sm-kraus-invariant-support}, every invariant support
\(S\) is invariant under
\(A_0,A_1,\ldots,A_n\).
Since
\[
A_0=P_E+\sqrt{1-\gamma}P_F
\]
has two distinct eigenvalues, its spectral projectors \(P_E,P_F\) are
polynomials in \(A_0\).
Hence \(S\) is invariant under both and therefore
\[
S=(S\cap E)\oplus(S\cap F).
\]

Suppose first that
\(S\cap F\neq\{0\}\).
Take
\[
v=\sum_jy_jf_j\in S\cap F,
\qquad
v\neq0.
\]
For every \(j\) with \(y_j\neq0\),
\[
A_jv
=
\sqrt{\gamma}\,y_je_j
\in S.
\]

For aquaternionicity,
\[
\langle\Theta e_j|v\rangle
=
\langle f_j|v\rangle
=
y_j\neq0,
\]
contradicting
\(S\perp\Theta S\).
Thus every invariant alchemical support satisfies
\(S\subseteq E\).
The full subspace \(E\) obeys
\(E\perp\Theta E\), so
\[
r_{\rm Aq}=n.
\]

For parity asymmetry, the decomposition
\[
S=(S\cap E)\oplus(S\cap F)
\]
is incompatible with
\(S\perp ZS\):
every nonzero vector in \(E\) has positive \(Z\)-norm, while every
nonzero vector in \(F\) has negative \(Z\)-norm.
Hence no nonzero invariant alchemical support exists and
\(r_{\rm Par}=0\).

For imaginarity, if
\(S\cap F\neq0\), the jump operators generate some real basis vector
\(e_j\in S\), which violates
\(S\perp KS\).
Hence again
\(S\subseteq E\).

The remaining problem is therefore to find the maximal dimension of a
subspace
\(S\subseteq\mathbb C^n\)
that is totally isotropic for the nondegenerate symmetric bilinear
form
\[
B(x,y)=x^Ty.
\]
Since
\(S\subseteq S^{\perp_B}\),
\[
2\dim_{\mathbb C}S\le n.
\]
The bound is attained by
\[
S
=
\operatorname{span}_{\mathbb C}
\left\{
\frac{e_1+ie_2}{\sqrt2},
\frac{e_3+ie_4}{\sqrt2},
\ldots
\right\},
\]
with
\(\lfloor n/2\rfloor\)
generators.
Therefore
\[
r_{\rm Im}
=
\left\lfloor\frac n2\right\rfloor.
\]
\end{proof}

For \(d=4\) (\(n=2\)),
\[
(r_{\rm Im},r_{\rm Par},r_{\rm Aq})
=
(1,0,2).
\]

\subsection{No universal ordering for arbitrary noise}

The preceding examples demonstrate an advantage of aquaternionicity
for two standard noise structures.
They do not imply that aquaternionicity dominates every other
alchemical theory for every possible CPTP noise channel.

\begin{proposition}[No universal ordering under arbitrary CPTP noise]
\label{prop:sm-no-universal-ordering}
In \(d=4\), there exist replacement channels for which the maximal
nonzero noise-invariant alchemical support exists only for
aquaternionicity, only for parity asymmetry, or only for imaginarity,
respectively.
Hence no universal ordering of the three theories holds for arbitrary
CPTP noise.
\end{proposition}

\begin{proof}
Use the decomposition
\[
\mathcal H=E\oplus F,
\qquad
E=\operatorname{span}\{e_1,e_2\},
\qquad
F=\operatorname{span}\{f_1,f_2\},
\]
with
\(Ze_j=e_j\),
\(Zf_j=-f_j\),
and
\(\Theta e_j=f_j\).

For aquaternionicity, take
\[
S_{\rm Aq}=E.
\]
Then
\(S_{\rm Aq}\perp\Theta S_{\rm Aq}\),
whereas
\(KS_{\rm Aq}=S_{\rm Aq}\) and
\(ZS_{\rm Aq}=S_{\rm Aq}\);
hence it is alchemical only for aquaternionicity.

For parity asymmetry, let
\[
U=
\begin{pmatrix}
0&1\\
-1&0
\end{pmatrix},
\qquad
S_{\rm Par}
=
\left\{
\frac{(x,Ux)}{\sqrt2}:
x\in\mathbb C^2
\right\}.
\]
Since \(U\) is unitary,
\(S_{\rm Par}\perp ZS_{\rm Par}\).
However,
\(U^T=-U\), so this support fails the aquaternionicity isotropy
condition; it also fails the imaginarity condition.

For imaginarity, take
\[
u
=
\frac{e_1+ie_2}{\sqrt2},
\qquad
v
=
\frac{f_1-if_2}{\sqrt2},
\qquad
S_{\rm Im}
=
\operatorname{span}\{u,v\}.
\]
Then
\[
u^Tu=v^Tv=u^Tv=0,
\]
so
\(S_{\rm Im}\perp KS_{\rm Im}\).
It is not parity-alchemical because \(u\) lies entirely in the
positive \(Z\)-sector, and it is not aquaternionicity-alchemical
because
\[
\Theta u=v
\]
and hence
\(\langle\Theta u|v\rangle=1\).

For each \(\nu\), let the noise channel replace every input by a
{state of full rank on \(S_\nu\).}
The subspace \(S_\nu\) is then an invariant alchemical support in the
corresponding theory.
If an alchemical support of either of the other theories contained
\(S_\nu\), then its subspace \(S_\nu\) would itself satisfy that
theory's support-orthogonality condition, contrary to the construction.
Thus no nonzero invariant alchemical support exists for the other two
theories.
\end{proof}

Accordingly, our comparison concerns structured noise classes rather
than an absolute ordering of resource theories under arbitrary noise.
For the standard dephasing and dissipative models analyzed above, the
absence of pure-state directional restrictions and the Kramers
isotropy of aquaternionicity lead to the separations reported in the
main text.